\documentclass[12pt]{article}
\usepackage{amssymb}

\usepackage[utf8]{inputenc}
\usepackage{tikz}
\usetikzlibrary{matrix,arrows,decorations.pathmorphing}
\usetikzlibrary{decorations.pathreplacing} 
\usetikzlibrary{patterns}
\usepackage{hyperref}
\usepackage[utf8]{inputenc}
\usepackage{subcaption}
\usepackage[english]{babel}
\usepackage[margin=1in]{geometry}
\usepackage{booktabs}
\usepackage{amsmath}
\usepackage{color,soul}
\usepackage{amsthm}
\usepackage{wrapfig}
\usepackage{bm}
\usepackage{mathtools}
\usepackage{caption}
\usepackage{lmodern,wrapfig,amsmath}
\usepackage{enumitem}
\usepackage{ stmaryrd }
\usetikzlibrary{shapes.misc}
\usetikzlibrary{cd}
\usepackage{adjustbox}
\usepackage{pgfplots}
\usepackage{tikz-3dplot}
\tdplotsetmaincoords{60}{115}
\pgfplotsset{compat=newest}
\usepackage{ amssymb }
\usepackage[backend=biber, style=ieee]{biblatex}
\bibliography{bibliography}
\usepackage{centernot}
\usepackage{faktor}
\usepackage{amsmath,amsfonts,amssymb, color, braket}

\tikzcdset{scale cd/.style={every label/.append style={scale=#1},
    cells={nodes={scale=#1}}}}

\tikzset{
  curarrow/.style={
  rounded corners=8pt,
  execute at begin to={every node/.style={fill=red}},
    to path={-- ([xshift=50pt]\tikztostart.center)
    |- (#1) node[fill=white] {$\scriptstyle \delta^*$}
    -| ([xshift=-50pt]\tikztotarget.center)
    -- (\tikztotarget)}
    }
}

\newtheorem{theorem}{Theorem}[section]

\newtheorem*{theorem*}{Theorem}

\newtheorem{lemma}[theorem]{Lemma}

\title{Topology Change from Pointlike Sources}
\author{Yasha Neiman\thanks{Email: yaakov.neiman@oist.jp} ~and David O'Connell\thanks{Email: david.oconnell@oist.jp}}
\date{\small{\textit{Okinawa Institute of Science and Technology,}} \\
\textit{1919-1 Tancha, Onna-son, Okinawa 904-0495, Japan} \\ \vspace{2mm}
\large{\today }}

\begin{document}
\maketitle 



\begin{abstract}
In this paper we study topology-changing spacetimes occurring from pointlike sources. Following an old idea of Penrose, we will opt for a non-Hausdorff model of topology change in which an initial pointlike source is ``doubled" and allowed to propagate along null rays into an eventual cobordism. By appealing to recent developments in non-Hausdorff differential geometry, we will describe and evaluate gravitational actions on these topology-changing spacetimes. Motivated by analogous results for the Trousers space, we describe a sign convention for Lorentzian angles that will ensure the dampening of our non-Hausdorff topology-changing spacetimes within a two-dimensional path integral for gravity. 
\end{abstract}

\newpage

\tableofcontents

\newpage

\section{Introduction}

In this paper we will concern ourselves with topology changing spacetimes and their transition amplitudes within a path integral for gravity. We consider perhaps the simplest non-trivial setting, that is, a transition from one circle into two.\footnote{Our reason for this will become clear throughout the paper: we will consider some unconventional types of topology change, and their novelty should already manifest in the two-dimensional regime where gravity is famously topological.} The customary model for this type of topology change is the so-called \textit{trousers space}, which is a smooth two-dimensional manifold that traces out the splitting process, as pictured in Figure \ref{FIG: trousers space}. The trousers space has long served as the prototypical example of topology change, and has been discussed in various physical contexts \cite{louko1997complex, sorkin1997forks, dowker2003topology, dowker1998handlebody, Dowker:1997hj, harris1990causal, de2022frontiers, borde1999causal, hartle1998generalized, sorkin1989consequences, dray1991particle, anderson1986does, manogue1988trousers, buck2017sorkin, witten2022note, feng2024singularity, Sorkin:1985bh}.  \\ 

Despite a broad discussion of topology change within physics, there exists an interesting gap in the literature that has remained unexplored for over half a century. When discussing time-asymmetry in \cite{penrose1979singularities}, Penrose sketches a particular type of topology-change that is markedly different from the trousers space. In his image, the manifold does not change its topology at a single point in time, but at a single point in \textit{spacetime}. This means that a single point changes its topology from one connected component to two, and then this change is allowed to grow along null rays. If taken within a compact universe, this will develop into a full topology change of spacelike slices within finite time, and thus may be used to model a possible transition between $S^1$ and $S^1 \sqcup S^1$. 
\\

Although a potentially interesting model for topology change, Penrose correctly identifies an important technical issue in his spacetimes: if one wants the pointlike splitting to remain a manifold, then models such as those pictured in Figure \ref{FIG: Penrose spacetime} are necessarily non-Hausdorff. At that time it was not clear how expressive a non-Hausdorff differential geometry could be, and thus after musing for some time and deviating wildly from the original scope of his paper, he finishes his discussion with the now-famous quote:\footnote{We have decided to include it in full, since the latter half of this quote is often misquoted as an argument against non-Hausdorff manifolds.}

\begin{quote}
    I have, in any case, strayed far too long from my avowed conventionality in this discussion, and no new insights as to the origin of time-asymmetry have, in any case, been obtained. I must therefore return firmly to sanity by repeating to myself three times: ‘spacetime is a \textit{Hausdorff} differentiable manifold; spacetime is a \textit{Hausdorff}...’
\end{quote}
In the past half century a theory of non-Hausdorff manifolds has emerged \cite{o2023nonHausAdj, o2023vectorbun, o2023deRham, buss2012non, crainic1999remark, francis2023h}, and these discussions have occasionally extended into non-Hausdorff spacetimes \cite{hajicek1970extensions, hajicek1971causality, luc2020interpreting, mccabe2005topology, heller2011geometry, muller2013generalized}. In this paper we will leverage recent mathematical developments in order to take seriously Penrose's splitting spacetimes. As a guiding reference, we will mimic the standard analysis of the trousers space found in \cite{louko1997complex, sorkin2019lorentzian} by defining and evaluating the gravitational action for a non-Hausdorff version of the trousers space. We will then compare transition amplitudes for the Hausdorff and non-Hausdorff trousers spaces, leading to the eventual conclusion that the latter admits a similar imaginary-strength action that yields a dampening  in the resulting path integral. \\

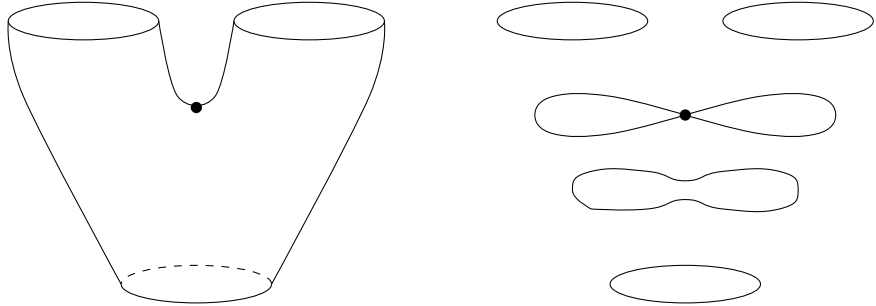
\begin{figure}
    \centering
    \begin{tikzpicture}[scale=0.5]

\draw (-3,7) ellipse (20mm and 5mm);
\draw (3,7) ellipse (20mm and 5mm);

\draw (2,0) arc
	[start angle=360,	end angle=180, x radius=20mm, y radius =5mm] ;
\draw[dashed] (2,0) arc
	[start angle=0,	end angle=180, x radius=20mm, y radius =5mm] ;

\draw[] plot[smooth, tension=0.6] coordinates{(-5,7) (-4.6,5) (-2,0)};
\draw[] plot[smooth, tension=0.6] coordinates{(5,7) (4.6,5) (2,0)};

\draw[] plot[smooth, tension=0.6] coordinates{(-1,7) (-0.5,5) (0.5,5) (1,7)};

\fill[] (0,4.7) circle (1.5mm);


\draw (10,7) ellipse (20mm and 5mm);
\draw (16,7) ellipse (20mm and 5mm);

\draw[] plot[smooth, tension=0.7] coordinates{(13,4.5) (15, 5) (16.5,5) (17,4.5) (16.5,4) (15,4) (13,4.5)};
\draw[] plot[smooth, tension=0.7] coordinates{(13,4.5) (11, 5) (9.5,5) (9,4.5) (9.5,4) (11,4) (13,4.5)};
\fill[] (13,4.5) circle (1.5mm);

\draw[] plot[smooth, tension=0.9] coordinates{(10.5,2) (12,2) (13,2.25) (14,2) (15.5, 2) (16, 2.5) (15.5, 3) (14,3) (13, 2.75) (12, 3) (10.5,3) (10,2.5) (10.5,2)  };

\draw (13,0) ellipse (20mm and 5mm);
    
    \end{tikzpicture}
    \caption{The trousers space, together with a sample of its spacelike slices. } 
    \label{FIG: trousers space}  
\end{figure}

In the remainder of this introduction we will briefly review some details of Hausdorff topology change and provide an informal discussion of non-Hausdorff topology, before outlining the paper in detail.

\subsection{Topology Change and Path Integrals}
An $n$-dimensional spacetime $M$ with boundary $\partial M = \Sigma_1 \sqcup \Sigma_2$ exhibits \textit{topology change} whenever the initial boundary $\Sigma_1$ and the final boundary $\Sigma_2$ are not homeomorphic. In such a model, the initial data $\Sigma_1$ is assumed to smoothly evolve and change its topology through time. Mathematically speaking, such interpolating manifolds are known as \textit{cobordisms}, and have been studied extensively in the literature.  \\

Following \cite{hawking1978quantum, dowker2003topology, sorkin1997forks}, in a naive sum-over-histories approach to quantum gravity we consider a path integral whose sum may include multiple distinct geometries and topologies. For topology change in Lorentzian signature, the transition amplitude between a pair of non-homeomorphic spacelike hypersurfaces $(\Sigma_1, h_1)$ and $(\Sigma_2, h_2)$ could be represented symbolically as 
    \begin{equation}\label{EQ: schematic Lorentzian path integral}
      \langle \Sigma_1, h_1 | \Sigma_2, h_2 \rangle =   \sum_M \int \mathcal{D}[g]\exp\{ i\mathcal{S}(M, g) \}, \footnote{Here we assume $\hbar= 1$.}
    \end{equation}
where here we sum over all physically-reasonable Lorentz cobordisms interpolating between $\Sigma_1$ and $\Sigma_2$. With a naive prescription such as the above, we are met with an immediate principled question: which interpolating manifolds should we sum over in the domain of the path integral? \\

\begin{figure}
    \centering
 \begin{tikzpicture}[scale=0.8]

\fill[top color=white, bottom color=gray!50, shading angle=45] (2.5,-0.5)--(0,2.5)--(-0.5,2.6)--(-0.5,-2.4)--(5.5,-3.6)--(5.5,1.4)--(5,1.5)--(2.5,-0.5);

\fill[top color=white, bottom color=gray!50, shading angle=45] plot[smooth, tension=0.8] coordinates {(0, 2.5) (0.25, 2.42) (1,2.15) (2,1.9) (3, 1.7) (4, 1.6) (5,1.5) } -- (2.5,-0.5)--(0,2.5);

\fill[top color=white, bottom color=gray!50, shading angle=45] plot[smooth, tension=0.8] coordinates {(0, 2.5) (0.25, 2.42) (1,2.15) (2,1.9) (3, 1.7) (4, 1.6) (5,1.5) } -- plot[smooth, tension=0.8] coordinates { (5,1.5) (4.45, 1.68) (3.8,1.9) (3,2.1) (2,2.3)  (1,2.4) (0.25, 2.475) (0, 2.5)  };

\draw[->] (-0.5,0.1)--(5.475,-1.09);
\draw[] (2.5,-3)--(2.5,-0.5);
\draw[] (0,2.5)--(-0.5,2.6)--(-0.5,-2.4)--(5.5,-3.6)--(5.5,1.4)--(5,1.5);

\draw[] (0,-2.5)--(2.5,-0.5)--(5,1.5);
\draw[] (5,-3.5)--(2.5,-0.5)--(0,2.5);

\draw[] plot[smooth, tension=0.8] coordinates {(0, 2.5) (0.25, 2.42) (1,2.15) (2,1.9) (3, 1.7) (4, 1.6) (5,1.5) };
\draw[] plot[smooth, tension=0.8] coordinates { (0, 2.5) (0.25, 2.475) (1,2.4) (2,2.3) (3,2.1) (3.8,1.9) (4.45, 1.68) (5,1.5) };

\draw[->] (2.5,-0.5)--(2.25,1.8);
\draw[opacity=0.3] (2.5,-0.5)--(2.6,1.77);
\draw[<-] (2.62,2.15)--(2.6,1.77);

\fill[](2.5,-0.5) circle[radius=0.05cm];

\draw[] (9.5,1.5)--(10.4,1.5);
\draw[] (13.6,1.5)--(14.5,1.5);
\fill[](10.5,1.6) circle[radius=0.05cm];
\fill[](10.5,1.4) circle[radius=0.05cm];
\fill[](13.5,1.6) circle[radius=0.05cm];
\fill[](13.5,1.4) circle[radius=0.05cm];
\draw[] (10.5,1.6)--(13.5, 1.6);
\draw[] (10.5,1.4)--(13.5, 1.4);

\draw[] (9.5,0.5)--(11.15,0.5);
\draw[] (12.85,0.5)--(14.5,0.5);
\fill[](11.25,0.6) circle[radius=0.05cm];
\fill[](11.25,0.4) circle[radius=0.05cm];
\fill[](12.75,0.6) circle[radius=0.05cm];
\fill[](12.75,0.4) circle[radius=0.05cm];
\draw[] (11.25,0.6)--(12.75,0.6);
\draw[] (11.25,0.4)--(12.75,0.4);

\draw[] (12.1,-0.5)--(14.5,-0.5);
\draw[] (9.5,-0.5)--(11.9,-0.5);
\fill[](12,-0.4) circle[radius=0.05cm];
\fill[](12,-0.6) circle[radius=0.05cm];

\draw[] (9.5,-1.5)--(14.5,-1.5);
\draw[] (9.5,-2.5)--(14.5,-2.5);

\draw[->] (16,-1.5)--(16,1.5);
\node[] at (16.3,1.5) {$t$};
\end{tikzpicture}
    \caption{A simplification of Penrose's non-Hausdorff topology changing spacetime (left), together with a sample of its level sets.
    Although impossible to depict in ordinary Euclidean space, here there are two copies of the origin superimposed on top of each other, together with two copies of the future null cone.}
    \label{FIG: Penrose spacetime}
\end{figure}
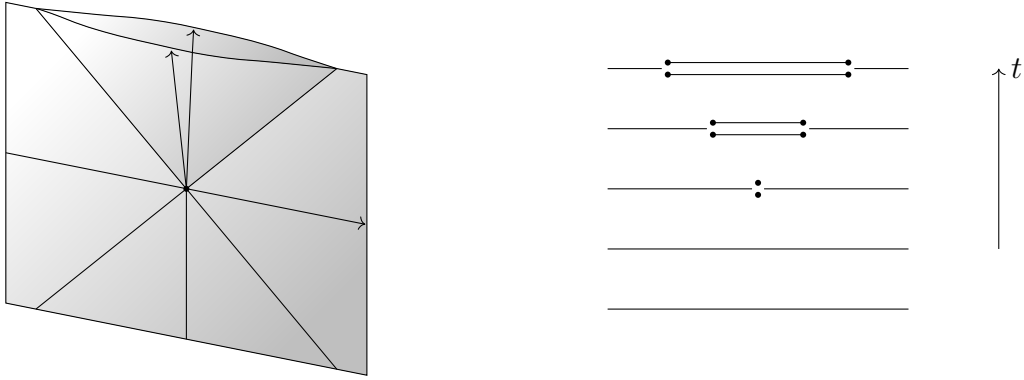

Before considering such an question, it makes sense to determine whether or not the sum is non-empty in the first place. In general dimensions, there are several known topological obstructions to the existence of Lorentz cobordisms, which are usually articulated as relationships between $\Sigma_1$ and $\Sigma_2$. It is well-known that any pair of $(n-1)$-dimensional manifolds will permit a smooth $n$-dimensional cobordism provided that they are related by a procedure known as \textit{Morse surgery} \cite{milnor2015lectures, reinhart1963cobordism}. Moreover, any smooth compact manifold admits a well-defined Lorentzian metric provided that it admits a globally non-vanishing vector field \cite{o1983semi}. This pair of observations may be taken in combination to explicitly describe Lorentz cobordisms via Morse theory \cite{yodzis1972lorentz, yodzis1973lorentz, alty1995building, dowker1998handlebody, Dowker:1997hj}. Note that in particular, Lorentz cobordisms exist between any pair of manifolds with $dim(\Sigma_i)\leq 3$. \\

After verifying the existence of Lorentz cobordisms, the next step might be to impose some top-down causal desiderata for the types of spacetimes we are willing to summing over. The exact specifications of this causal behaviour are still subject to debate, but at the very least it seems as though our spacetimes ought to admit a global time function, since this is implicitly used in the formation of \eqref{EQ: schematic Lorentzian path integral}. The existence of global time functions is known to be equivalent to the property of \textit{stable causality} \cite{hawking2023large}. Roughly put, a spacetime is stably causal provided it contains no closed timelike curves, and neither does any small perturbation of its metric. Such spacetimes lie relatively high up in the causal hierarchy of \cite{minguzzi2008causal, minguzzi2019lorentzian}. \\

At the apex of this causal hierarchy are the \textit{globally hyperbolic} spacetimes. These spacetimes are just about as causally well-behaved as possible, in that they always admit Cauchy surfaces. However, it can be shown that any globally hyperbolic spacetime is necessarily cylindrical, in that both its topological and smooth manifold structure are isomorphic (in the appropriate sense) to a product of the Cauchy surface with either the real line or the unit interval \cite{geroch1970domain, bernal2003smooth}. As such, we see that topology-changing spacetimes may never be globally-hyperbolic. Nonetheless, it appears as though we may be able to include topology-changing spacetimes within \eqref{EQ: schematic Lorentzian path integral} if relax our causality requirements slightly and allow for stably-causal spacetimes. However, we are then met with the following result of Geroch (adapted from \cite{geroch1967topology}). 
\begin{theorem*}[Geroch]\label{THM: geroch CTCs}
    Let $M$ be a compact spacetime with initial boundary $\Sigma_1$ and final boundary $\Sigma_2$. If $\Sigma_1$ and $\Sigma_2$ are not homeomorphic, then $M$ admits a closed timelike curve. 
\end{theorem*}

We are thus prompted to exclude topology-changing spacetimes on causal grounds, since they are not stably causal. A possible circumvention of this issue is to relax the assumption that $M$ admits a globally-defined Lorentzian metric. Instead, we may allow an \textit{almost} Lorentzian manifold in which the metric is allowed to degenerate at select points in the space. Within the context of topology change, this approach seems quite reasonable, as it naturally aligns with the Morse-theoretic view of smooth cobordisms. Under this reading, we may consider what is known as a \textit{Morse function}, which has critical points along which the topology of space will change. The gradient of this Morse function will then provide us with a vector field that vanishes only on the critical points, and this vector field may be used to define an almost-Lorentzian manifold. \\

Another causal requirement might be to suggest that the almost-Lorentz cobordism induces a \textit{causal poset} structure on its lightcones. According to the causality theory of \cite{penrose1972techniques, dowker2003topology, minguzzi2008causal, minguzzi2019lorentzian}, we may induce a binary \textit{causal precedence} relation $\leq$ on any almost-Lorentzian manifold. This relation states a point $p$ \textit{causally precedes} a point $q$, written $p \leq q$ if $q$ lies in the future lightcone of $p$. Causality properties may then be paraphrased as order-theoretic properties of the relation $\leq$. In particular, one may suggest including into \eqref{EQ: schematic Lorentzian path integral} all almost-Lorentz cobordisms in which the binary relation $\leq$ is reflexive, transitive and antisymmetric \cite{dowker2003topology}. Note that the requirement of antisymmetry excludes those spacetimes admitting closed timelike curves.

\subsection{The Trousers Space}
Topologically, the trousers space is homeomorphic to the $3$-punctured sphere, and may be seen as a cobordism from $S^1$ to $S^1\sqcup S^1$. This manifold is almost-Lorentzian, in the sense that it admits a non-degenerate Lorentzian metric everywhere except at a single point (commonly called the \textit{crotch singularity}). Away from this point, the Trousers space may be furnished with a Lorentzian metric that is locally isometric to the flat cylinder \cite{louko1997complex, anderson1986does}. \\

It seems reasonable to suggests that one may avoid Geroch's theorem by simply removing this troublesome point and allowing the manifold to be non-compact. However, with transition amplitudes such as \eqref{EQ: schematic Lorentzian path integral} in mind, it seems that we would like to preserve compactness as best as we can. An alternate resolution involves what is known as a \textit{causal closure} construction \cite{gerpenkron1972ideal}. In this method one chooses to maintain compactness, at the cost of an allowed degeneracy in the metric at the crotch singularity. A lightcone structure may still be placed at the crotch singularity, however this structure will be irregular in the sense that there will be double the amount of distinct lightcones -- two future-directed and two past-directed \cite{harris1990causal}. Suggestive depictions of the Trousers space and the origin of the irregular causal structure of its crotch singularity can be found in Figure 2 of the recent paper \cite{feng2024singularity}. \\

With a causally-closed trousers space, we may still evaluate the transition amplitude of \eqref{EQ: schematic Lorentzian path integral}. In two dimensions, the gravitational action in vacuum is given by the total scalar curvature of the manifold: 
    $$  \mathcal{S}(M,g) = \frac{1}{2\kappa}\int_M R dA +\frac{1}{\kappa} \int_{\partial M} k d\gamma,       $$
where here the latter term computes the geodesic curvature of the boundary and $\kappa$ is the gravitational constant. An analysis of Louko and Sorkin shows that the above action may still be evaluated for the Trousers space \cite{louko1997complex}. In their work they employ an $i\epsilon$-regularisation in which the Lorentzian metric is perturbed into a complex one in a controlled manner. Their conclusion is that the Trousers space has a $\delta$-like curvature localised to the crotch singularity of strength $\pm 2\pi i$, with the ambiguity being controlled by the sign of the regulariser. \\

Alternatively, the curvature of the Trousers space may be obtained via a certain Lorentzian Gauss-Bonnet theorem \cite{sorkin2019lorentzian}. We will discuss the details of the Lorentzian Gauss-Bonnet theorem and its subtleties at length in Section 3 of this paper, so for now we will merely deliver a brief summary. In short: the notion of Lorentzian angle does not make sense for vectors of different signatures, again it is commonplace to employ the $i\epsilon$-regularisation of \cite{louko1997complex} and complexify the Minkowski metric. This gives a notion of Lorentzian angle that is well-defined, yet complex-valued \cite{sorkin2019lorentzian, asante2023complex, neiman2013imaginary, jia2022complex}. The Lorentzian Gauss-Bonnet theorem then loosely states that: 
    $$\frac{1}{2}\int_M R dA + \int_{\partial M} k d\gamma = \mp 2\pi i \chi(M). $$
The imaginary coefficient on the right-hand-side arises from our complexification of the metric, with $\mp 2\pi i$ being the periodicity of angles around a point in flat two-dimensional Minkwoski space, and the sign ambiguity again provided by the $i\epsilon$-regulariser. For the trousers space the Euler characteristic equals $-1$, so we may conclude that its total scalar curvature equals $\pm 2\pi i$, in agreement with the prior analysis of \cite{louko1997complex}. \\

The sign ambiguity in the $i\epsilon$-regulariser may be resolved in several ways. It is commonly argued that the correct sign of the imaginary periodicity should be $-2\pi i$, since then 
$$  \exp \{i\mathcal{S}(M,g)\} =  \exp \left\{i\left( \frac{-2\pi i \chi(M)}{\kappa}\right) \right\} =  \exp \left\{\frac{i(+2\pi i)}{\kappa}\right\} =  \exp \left\{\frac{-2\pi}{\kappa}\right\} < 1, $$
which would in turn cause the trousers space to be suppressed relative to other spacetimes like the cylinder or the $2$-disk \cite{louko1997complex, sorkin2019lorentzian, asante2023complex, neiman2013imaginary}. In this paper we will only focus on transitions from $S^1$ to $S^1 \sqcup S^1$. Within our context, a similar argument is valid: creating some more complicated trousers-like transition between circles by adding extra genera to the bulk will always decrease the Euler characteristic. The Lorentzian Gauss-Bonnet theorem roughly stated above then implies that adding more and more genera to the Trousers space will continually increase the imaginary part of the action, and thus the transition amplitudes of \eqref{EQ: schematic Lorentzian path integral} will become exponentially small. In contrast, if we were to choose the other sign convention, then we are left with a periodicity of $+2\pi i$ in which higher genus trousers-like spaces would be exponentially enhanced. In this sense, the sign convention advocated in \cite{louko1997complex, sorkin2019lorentzian, asante2023complex, neiman2013imaginary} indeed appears to be the correct one. 

\subsection{A Primer on non-Hausdorff Topologies}
The Hausdorff property states that any pair of distinct points in a topological space may be separated by disjoint open sets. Conversely, a topological space is called \textit{non-Hausdorff} whenever there exists a pair of points whose open neighbourhoods always intersect. Hausdorffness is usually assumed in the definition of a manifold, however there still exist non-Hausdorff locally-Euclidean spaces. For example, the right-hand-side of Figure \ref{fig: nh manifolds} depicts a simple one dimensional non-Hausdorff manifold, commonly called the \textit{branched line}. In this space there are two copies of the origin, and the topology is defined to be locally-equivalent to the real line. One can see that any pair of open intervals around the Hausdorff-violating pair of origins will necessarily intersect throughout the negative numbers. \\

    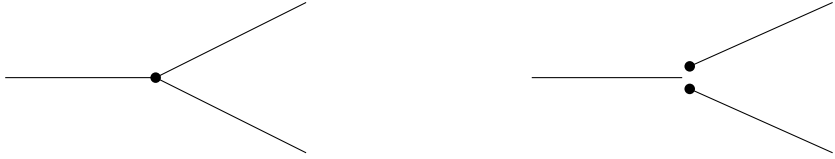
\begin{figure}
        \centering
        \begin{tikzpicture}
        \draw[] (-2,0)--(0,0);
        \draw[] (0,0)--(2,1);
        \draw[] (0,0)--(2,-1);
        \fill[] (0,0) circle (0.7mm);

        \draw[] (5,0)--(7,0);
        \draw[] (7.1,0.15)--(9,1);
        \draw[] (7.1,-0.15)--(9,-1);
        \fill[] (7.1,0.15) circle (0.7mm);
        \fill[] (7.1,-0.15) circle (0.7mm);
        
    \end{tikzpicture}
        \caption{Hausdorff (left) and non-Hausdorff (right) topology change in one dimension. The Hausdorff topology change is necessarily singular, whereas the non-Hausdorff model blows up the singularity into two distinct points and adjusts their separability in order to ensure that the space remains locally-Euclidean.}
        \label{fig: nh manifolds}
    \end{figure}

In the manifold setting, Hausdorff violation may be seen in many different ways. Perhaps the most instructive is via the non-uniqueness of limits: it is well-known that any convergent sequence in a manifold has a unique limit, provided that manifold is Hausdorff. In the non-Hausdorff setting this is no longer true -- in fact, a pair of points will violate the Hausdorff property if and only if they can be realised as distinct limits of the same sequence \cite{o2023nonHausAdj}.
\\ 

Usually we include the Hausdorff property in the definition of a manifold for technical convenience. In particular, it can be shown that any open cover of a Hausdorff manifold admits a partition of unity subordinate to it \cite{lee2013smooth}. These arbitrarily-existent partitions of unity are used frequently in order to construct various geometric structures of interest. In the non-Hausdorff case, such partitions of unity do not exist in full generality, and thus the various enjoyable features of Hausdorff manifolds may appear to be in jeopardy. Put differently: without the Hausdorff property we do not have access to the usual constructive techniques, and \textit{prima facie} it is not clear whether non-Hausdorff manifolds are as expressive as their Hausdorff counterparts. \\

Despite this issue with partitions of unity, it is nonetheless possible to describe a differential geometry of non-Hausdorff manifolds. Underpinning this study is the observation that non-Hausdorff manifolds may be constructed by gluing together ordinary Hausdorff manifolds along open sets \cite{haefliger1957varietes, hajicek1971causality, o2023nonHausAdj, luc2020interpreting}. Intuitively, if one glues together Hausdorff manifolds along an open subset but leaves the boundary of this subset unidentified, then this boundary may become Hausdorff-violating in the quotient space. As an example: we may realise the branched line of Figure \ref{fig: nh manifolds} by gluing together two copies of the real line along the subset $A:=(-\infty, 0)$. In the resulting quotient space, any sequence of negative numbers that would ordinarily converge to the origin will now have two distinct limits, thereby realising Hausdorff-violation. 

\subsection{Outline of paper}
In Section 2 we will provide a formal overview of non-Hausdorff differential geometry. We will start with matters topological, and then move on to smooth structures, bundles, integrals and eventually curvature. Underpinning our discussion is the aforementioned gluing concept -- essentially all of these geometric structures may be defined on non-Hausdorff spaces by first defining them on Hausdorff submanifolds, and then by imposing some consistency conditions on overlapping submanifolds. The non-Hausdorff manifolds that we define will be locally-isomorphic to Hausdorff ones, however their global features will differ. In particular, we will see that their notion of integration needs to include the extra Hausdorff-violating data in order to be well-defined, and this in turn will have some far-reaching consequences for the rest of the paper.  \\

In Section 3 we will discuss various extensions of the Gauss-Bonnet theorem. We start with the standard statement for Riemannian manifolds found in say \cite{do2016differential}, and we will then modify it in two orthogonal directions: firstly, we will pass from Riemannian metrics to Lorentzian metrics, and secondly, we will pass from Hausdorff surfaces to non-Hausdorff ones. The result of our discussion will be a non-Hausdorff version of the  Gauss-Bonnet theorem in Lorentzian signature. Here we will see a crucial novelty -- due to integration results of Section 2, the non-Hausdorff result will require an extra counterterm that computes the geodesic curvature of the Hausdorff-violating submanifold sitting inside the manifold. This counterterm is a sort-of ``internal boundary" term that has no analogue in the Hausdorff regime. \\

In Section 4 we provide the primary contribution of this paper. Here, we will study a non-Hausdorff version of the trousers space. In essence, this ``non-Hausdorff trousers space" can be see as a version of Penrose's spacetime of Figure \ref{FIG: Penrose spacetime} that has been compactified so that its initial and final surfaces equal $S^1$ and $S^1 \sqcup S^1$, respectively. To begin with, in Section 4.1 we will analyse the causal properties of the non-Hausdorff trousers space. Using the results of Section 2, we will argue that this space cannot be excluded from the path integral \eqref{EQ: schematic Lorentzian path integral} on the basis that it is not a rich-enough geometric structure. We will then analyse its causal properties by confirming the existence of a global time function, the non-existence of closed timelike curves, the compactness of its causal diamonds, and the poset structure of its causality relation $\leq$. \\

In the remainder of Section 4 we will then determine the gravitational action for the non-Hausdorff trousers space. As an organisational choice, we will first motivate the Lorentzian action from its Euclidean cousin. In line with the Gauss-Bonnet theorems of Section 3, we will see that the non-Hausdorff gravitational action requires another Gibbons-Hawking-York term for the extra Hausdorff-violating surface. We will argue that the non-Hausdorff trousers space has zero curvature, meaning that in the Euclidean theory there is no inherent mechanism that would enable its suppression. In the Lorentzian theory, however, we will see that the presence of corner terms in the action, together with the freedom to choose signs of the $i\epsilon$-regulator, will allow us to suppress the non-Hausdorff trousers space as desired. Finally, we finish with some brief remarks regarding more elaborate non-Hausdorff branching. 

\section{Non-Hausdorff Differential Geometry}
We start with a review of non-Hausdorff manifolds. Throughout this section and the remainder of this paper, we will reserve the term ``manifold" for its ordinary usage, that is, manifolds are taken to be Hausdorff, locally-Euclidean and second-countable. In distinction to this, we will use the term ``non-Hausdorff manifold" to mean a non-Hausdorff, locally-Euclidean, second-countable space.  Regarding notation, we will mostly follow the notation of Lee in \cite{lee2013smooth} for ordinary differential geometry, and \cite{o2023nonHausAdj, o2023vectorbun, o2023deRham} for non-Hausdorff variants. In particular, we will use boldface letters to denote non-Hausdorff manifolds and their geometric structures -- for example $\textbf{M,N},..$ would denote non-Hausdorff manifolds whereas $M,N,...$ would denote Hausdorff ones. 

\subsection{Topological Structure}
To begin, we will describe a general technique for gluing together manifolds. This formalism is known as an \textit{adjunction space} in the literature, though may take subtly different forms depending on the context. We assume as input two Hausdorff manifolds $M_1$ and $M_2$ of the same dimension, a subset $A$ of $M_1$, and a continuous map $f:A\rightarrow M_2$. We may then glue $M_1$ to $M_2$ along the map $f$ by quotienting the disjoint union $M_1 \sqcup M_2$ according to the equivalence relation that identifies each point in $A$ to its image under $f$. \\

This notion of adjunction space is far too general, and may spoil the topological structure of $A$ in the gluing process. The following result identifies some conditions under which the adjunction space described above will yield a non-Hausdorff manifold.

\begin{theorem}[{\cite{o2023nonHausAdj}}]\label{THM: topological nh manifold}
    Let $M_1$ and $M_2$ be two Hausdorff manifolds of the same dimension, and let $A$ be an open subset of $M_1$. Suppose that $f:A\rightarrow M_2$ is an open topological embedding. If $f$ can be extended to a closed embedding $f:\overline{A}\rightarrow M_2$,\footnote{Here we use the notation $\overline{A}$ to denote the topological closure of $A$ within $M_1$.} then the quotient space $$ \textbf{M} := \frac{M_1 \sqcup M_2}{a \sim f(a)} = M_1 \cup_f M_2$$ is a non-Hausdorff manifold in which the Hausdorff-violating points occur precisely at the boundary of the image of $A$ in the quotient space.  
\end{theorem}
At first glance, the above appears to be very similar to the connected sum of manifolds (cf. \cite{lee2013smooth}). However, there is an important distinction: the connected sum assumes that the gluing region $A$ is topologically \textit{closed}, which is necessary in order to preserve the Hausdorff property. However, in our context, we want to take our gluing region to be an \textit{open} subset with a non-empty boundary. This assumption intentionally spoils the Hausdorff property, since the boundaries of these glued open sets remain unidentified, thus serve as distinct limits to the same sequences.  \\

It should be noted at this point that the adjunction space construction may also be phrased categorically -- the quotient construction of $\textbf{M}$ in the above result can be viewed as the colimit of the upper-left corner of the diagram below in the category of topological spaces. 
\begin{center}
    \begin{tikzcd}[row sep=3.5em]
A \arrow[d, "f"'] \arrow[r, "\iota"] & M_1 \arrow[d, "\phi_1"] \\
M_2 \arrow[r, "\phi_2"']             & \textbf{M}             
\end{tikzcd}
\end{center}
In the above the map $\iota: A\rightarrow M_1$ is the inclusion map, and the $\phi_i: M_i \rightarrow \textbf{M}$ are the canonical maps that send each point to its equivalence class in $\textbf{M}$. It can be shown that the maps $\phi_i$ are open topological embeddings, provided that $f$ is open and $A$ itself is \cite{o2023nonHausAdj}, and in fact this is precisely what is used in order to transfer the local charts from $M_i$ into $\textbf{M}$. By construction, any local chart $(U, \varphi)$ of $M_i$ defines a chart $(\phi_i(U), \varphi\circ \phi_i^{-1})$ on $\textbf{M}$, and it is in this sense that the non-Hausdorff manifold $\textbf{M}$ of Theorem \ref{THM: topological nh manifold} is locally equivalent to the manifolds $M_1$ and $M_2$. Given that the $\phi_i$ maps are open topological embeddings, we may see $M_1$ and $M_2$ as sitting inside $\textbf{M}$ as maximal Hausdorff open submanifolds. \\

The idea that the maps $\phi_i$ will be as equally well-behaved as the gluing map $f$ may be extended beyond topology alone. As the next result illustrates, we may actually pass this entire adjunction construction into an appropriate smooth category.

\begin{theorem}[{\cite{o2023vectorbun}}]\label{THM: smooth nh manifold}
    Suppose in addition to the criteria of Theorem \ref{THM: topological nh manifold} that the $M_i$ and $A$ are all smooth manifolds, and $f:A\rightarrow M_2$ is a smooth map. If $f$ can be extended to a smooth embedding $\overline{f}:\overline{A}\rightarrow M_2$, then $\textbf{M}$ can be endowed with a smooth atlas.
\end{theorem}

Once endowed with a smooth atlas, the canonical embeddings $\phi_i: M_i \rightarrow \textbf{M}$ now become \textit{smooth} open embeddings. Consequently, we may view the Hausdorff manifolds $M_i$ as smooth open submanifolds of $\textbf{M}$, with the $\phi_i$ acting as local diffeomorphisms. \\

It should be noted at this stage that the colimit formulation of \cite{o2023nonHausAdj, o2023vectorbun} is far more general than what we have presented here, in that it may also be extended to colimits of more than two manifolds. However, since we will only be considering transitions from $S^1$ to $S^1\sqcup S^1$, we will not require this formalism in full generality. Throughout the remainder of this section, we will assume that $\textbf{M}$ is a non-Hausdorff manifold built according to Theorems \ref{THM: topological nh manifold} and \ref{THM: smooth nh manifold}.

\subsection{Vector Bundles}
Smooth vector bundles over a non-Hausdorff manifold can be described with an analogue of the colimit construction of Theorems \ref{THM: topological nh manifold} and \ref{THM: smooth nh manifold}. The only major difference is that we must also require the existence of a gluing map for the fibres of the part of the bundle that lies over the gluing region $A$. Once this is correctly done, we may indeed glue bundles along their fibres in order to form a non-Hausdorff vector bundle. In a manner similar to that of \cite{o2023nonHausAdj, luc2020interpreting}, a converse to this construction holds: any vector bundle $\textbf{E}$ fibred over $\textbf{M}$ is in fact a colimit of ordinary Hausdorff bundles $E_i$ fibred over the $M_i$. \\

Intuitively speaking, we can represent any smooth section $\textbf{s}$ of a non-Hausdorff bundle $\textbf{E} \xrightarrow{\boldsymbol{\pi}} \textbf{M}$ by pulling it back to the Hausdorff bundles $\phi_i^*\textbf{E} \rightarrow M_i$ and describing it piecewise. Provided that the two pulled-back sections $\phi_i^*\textbf{s}$ agree once mutually restricted to the gluing region $\textbf{A}$, it is then possible to canonically reconstruct $\textbf{s}$ from Hausdorff data. Figure \ref{FIG: gluing sections} depicts a semi-local representation of a non-Hausdorff section. \\

The pullback correspondence of Figure \ref{FIG: gluing sections} may also be phrased algebraically. For any non-Hausdorff vector bundle $\textbf{E}$ over $\textbf{M}$, we may use pointwise addition and multiplication by scalar functions to endow the space of sections $\Gamma(\textbf{E})$ with the structure of a $C^\infty(\textbf{M})$-module. This can then be related to the spaces of sections of the Hausdorff submanifolds as follows. 

\begin{theorem}[{\cite{o2023nonHausAdj}}] \label{THM: sections are fibre product}
    Let $\textbf{E}$ be a vector bundle over $\textbf{M}$ with colimit representation $E \cong E_1 \cup_{F}E_2$ where $F:E_A\rightarrow E_2$ is a bundle isomorphism covering $f:A\rightarrow M_2$. Then the space of smooth sections $\Gamma(\textbf{E})$ is canonically isomorphic to the fibre product: 
    $$ \Gamma(\textbf{E}) \cong \Gamma(E_1)\times_{\Gamma(E_A)} \Gamma(E_2) = \big\{ (s_1, s_2) \in \Gamma(E_1)\oplus \Gamma(E_2) \ \big{|} \  \iota_A^* s_1 = f^*s_2    \big\}.$$
\end{theorem}
On the level of an individual section, the above result is stating that defining a smooth section $\textbf{s}$ on $\textbf{E}$ amounts to defining a pair of sections $s_i$ on the restricted bundles $E_i$ that agree once mutually pulled back to the bundle $E_A$. On the categorical level, the pullback of sections by smooth maps allows us to see $\Gamma(\cdot)$ as a contravariant functor. Once applied to the colimit diagram used to construct/describe the bundle $\textbf{E}$, we obtain a new diagram in a particular (abelian) category of modules over rings. Theorem \ref{THM: sections are fibre product} then states that the contravariant functor $\Gamma(\cdot)$ sends our colimit $\textbf{E}$ into the limit $\Gamma(\textbf{E})$. \\
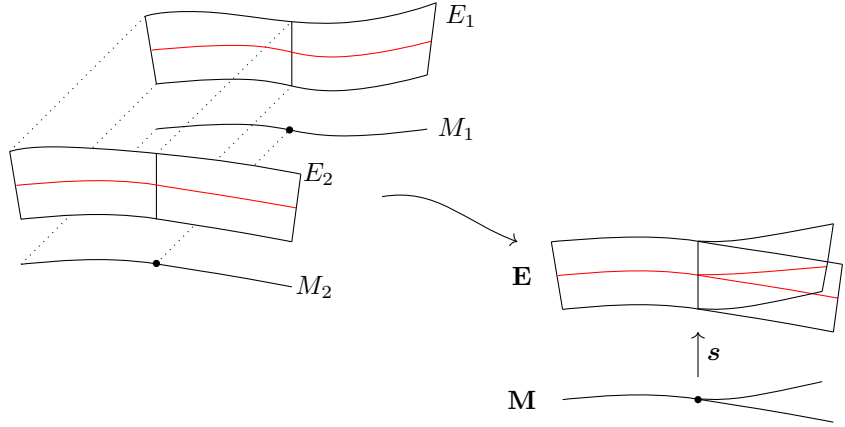
\begin{figure}
    \centering
    \begin{tikzpicture}[scale=0.6]

    \draw[dotted] (-9,3.02)--(-6.05,5.98);
    \draw[dotted] (-9,4)--(-6,6.95);
    \draw[dotted] (-9,5.45)--(-6,8.39);
    \draw[dotted] (-12,3)--(-9,6);
    \draw[dotted] (-12,4)--(-9,7);
    \draw[dotted] (-12.25,5.5)--(-9.25,8.5);
    
    
    \fill[white, opacity=0.8] plot[smooth, tension=.8] coordinates {(-6, 3.5) (-8,3.85) (-10,4.1) (-12,4)} -- plot[smooth, tension=.8] coordinates {(-12.25,5.5) (-11,5.6) (-8,5.35) (-5.8, 5)};
    
    \draw[->] (-4,4.5) to [out=10,in=160] (-1,3.5);
    \node[] at (-2.35,6) {\footnotesize{$M_1$}};
    \node[] at (-2.25,8.5) {\footnotesize{$E_1$}}; 
    \node[] at (-5.5,2.5) {\footnotesize{$M_2$}};
    \node[] at (-5.4,5) {\footnotesize{$E_2$}};

    \draw[] plot[smooth, tension=.8] coordinates {(-9,6) (-7,6.1) (-5,5.85) (-3, 6)};
    \draw[] plot[smooth, tension=.8] coordinates {(-9,7) (-7,7.1) (-5,6.85) (-3, 7.2)};
    \draw[] plot[smooth, tension=.8] coordinates {(-9.25,8.5) (-8,8.6) (-5,8.35) (-2.8, 8.8)};
    
    \draw[] (-9.25,8.5)--(-9,7);
    \draw[] (-2.8,8.8)--(-3,7.2);
    \draw[] (-6,8.39)--(-6,6.95);
    
    \fill(-6.05,5.98) circle[radius=0.08cm] {};
    


    \draw[] plot[smooth, tension=.8] coordinates {(-12,3) (-10,3.1) (-8,2.85) (-6, 2.5)};
    \draw[] plot[smooth, tension=.8] coordinates {(-12,4) (-10,4.1) (-8,3.85) (-6, 3.5)};
    \draw[] plot[smooth, tension=.8] coordinates {(-12.25,5.5) (-11,5.6) (-8,5.35) (-5.8, 5)};
    
    \draw[] (-12.25,5.5)--(-12,4);
    \draw[] (-5.8,5)--(-6,3.5);
    \draw[] (-9,5.45)--(-9,4);
    
    \fill(-9,3.02) circle[radius=0.08cm] {};
    

    \draw[red] plot[smooth, tension=.8] coordinates {(-9.1,7.75) (-7,7.85) (-5,7.6) (-2.9, 7.95)};
    \draw[red] plot[smooth, tension=.8] coordinates {(-12.1,4.75) (-10,4.85) (-8,4.6) (-5.9, 4.25)};
    
    \draw[red] plot[smooth, tension=.8] coordinates {(-0.1,2.75) (2,2.85) (4,2.6) (6.1, 2.25)};
    \draw[red, opacity=0.1] plot[smooth, tension=.8] coordinates {(3,2.77) (4,2.8) (5.85, 2.95)};
    
    \draw[] plot[smooth, tension=.8] coordinates {(0,0) (2,0.1) (4,-0.15) (6, -0.5)};
    \draw[] plot[smooth, tension=.8] coordinates {(3,0.02) (4,0.05) (5.75, 0.4)};
    \draw[] plot[smooth, tension=.8] coordinates {(0,2) (2,2.1) (4,1.85) (6, 1.5)};
    \draw[opacity=0.1] plot[smooth, tension=.8] coordinates {(3,2.02) (4,2.05) (5.75, 2.4)};
    \draw[] plot[smooth, tension=.8] coordinates {(-0.25,3.5) (2,3.6) (4,3.35) (6.2, 3)};
    \draw[] plot[smooth, tension=.8] coordinates {(3,3.52) (4,3.55) (6, 3.9)};
    
    \draw[] (-0.25,3.5)--(0,2);
    \draw[] (6.2,3)--(6,1.5);
    \draw[] (6, 3.9)--(5.87,3.04);
    \draw[opacity=0.1] (5.75, 2.4)--(5.87,3.04);
    
    \node[] at (-0.9,0) {\footnotesize{$\textbf{M}$}};
    \node[] at (-0.9,2.75) {\footnotesize{$\textbf{E}$}};    
    
    \fill(3,0) circle[radius=0.08cm] {};
    \draw[] (3,3.5)--(3,2);
    \draw[, <-] (3,1.5)--(3,0.5);
    \node[] at (3.35, 1) {\footnotesize{$\boldsymbol{s}$}};
     
\end{tikzpicture}  
    \caption{A semi-local depiction of the formation of sections of non-Hausdorff bundles.}
    \label{FIG: gluing sections}
\end{figure}
    
These abstract bundle-theoretic arguments can be applied to the tangent bundle $T\textbf{M}$ in order to describe the vector fields over $\textbf{M}$. To begin with, one can show that the tangent bundle $T\textbf{M}$ is canonically isomorphic to the colimit of the bundles $TM_1$ and $TM_2$, glued along $TA$ via the gluing (bundle) map $df: TA \rightarrow TM_2$. Since $f$ is an open embedding, its differential $df$ is a bundle embedding, and thus falls under the scope of Theorem \ref{THM: sections are fibre product}. We may then conclude that any vector field $\textbf{v}$ in $\Gamma(T\textbf{M})$ can be uniquely described by a pair of vector fields $v_i$ in $\Gamma(TM_i)$ that agree once restricted to $A$. The higher-rank tensorial bundles also admit a similar colimit construction (cf. \cite[Sec. 2]{o2023vectorbun}), and the tensor fields on the non-Hausdorff manifold $\textbf{M}$  may therefore be described with the fibre product formula of Theorem \ref{THM: sections are fibre product}.

\subsection{Integration}
In our discussions thus far we have been identifying conditions under which locally-defined data may be described in the non-Hausdorff case, ultimately by a transfer of the Hausdorff data under the canonical maps $\phi_i$. Once the correct gluing conditions were identified, our discussion was somewhat intuitive and unproblematic. However, despite being locally-equivalent to Hausdorff manifolds, there is a significant issue when passing from local to global structures. \\

In Hausdorff differential geometry, a useful method for pasting together locally-defined objects is via partitions of unity. The precise definition is not necessary for our purposes, but their utility should not be understated: partitions of unity are used in various constructions and arguments for manifolds, including the locality of derivative operators, the construction of Riemannian metrics and their Levi-Civita connections, and various arguments involving de Rham cohomology. In the non-Hausdorff case, we have the following inconvenient fact. 

\begin{theorem}[{\cite{o2023deRham}}]\label{THM: no partitions of unity}
    If $\textbf{M}$ be a non-Hausdorff manifold built according to Theorems \ref{THM: topological nh manifold} and \ref{THM: smooth nh manifold}, then there is no partition of unity subordinate to any cover of $\textbf{M}$ by Hausdorff open sets. 
\end{theorem}

In Hausdorff differential geometry, the integral of a compactly-supported differential form is performed by decomposing the form into several local charts, performing the integrals in Euclidean space, and then summing over the results via a partition of unity. In this context, the partition of unity is required in order to avoid overcounting the integral on overlapping charts. \\ 

In the non-Hausdorff setting we do not have access to partitions of unity in full generality, so we will need to define integration on a non-Hausdorff manifold by alternate means. Instead of appealing to partitions of unity, we will follow the so-called ``integration over parametrizations", found in say \cite[Chpt. 13]{lee2013smooth}. Roughly put, in this scheme a total integral is broken down into integrals over certain open sets whose closures cover the support of the differential form, in such a way that adjacent regions only intersect at their measure-zero boundaries. With this intersection property there is no risk of overcounting the integral, and thus partitions of unity are not required. \\

Suppose that we have some compactly-supported differential form $\boldsymbol{\omega}$ on $\textbf{M}$, and consider a collection $\{ r_i: U_i \rightarrow \mathbf{M} \}$ of finitely-many open domains of integration $U_i$ that are mapped diffeomorphically into $\textbf{M}$ such that orientations are preserved. Provided that the union $\bigcup_i r_i(U_i)$ cover the support of $\boldsymbol{\omega}$ and the sets $r_i(U_i)$ pairwise intersect on at most their boundaries in $\textbf{M}$, we may define the integral of $\boldsymbol{\omega}$ in $\textbf{M}$ to be $$  \int_{\textbf{M}} \boldsymbol{\omega} = \sum_i \int_{U_i} r_i^*\boldsymbol{\omega}.   $$
Although the above is technically a well-defined notion of integration, it is not particularly useful for our needs. What is more helpful for us is the following, which relates the integral of a form on $\textbf{M}$ to the ordinary integrals over the Hausdorff submanifolds $M_1$, $M_2$ and $A$.

\begin{theorem}[{\cite[Thm. 2.6]{o2023nonHausAdj}}]\label{THM: Integration formula}
    Let $\boldsymbol{\omega}$ be a compactly-supported differential form on $\textbf{M}$. Then $$  \int_{\textbf{M}} \boldsymbol \omega = \int_{M_1} \omega_1 + \int_{M_2} \omega_2 - \int_{\overline{A}} \omega_A,     $$ 
    where $\overline{A}$ is the topological closure of $A$ within $M_1$, and $\omega_i := \phi_i^* \boldsymbol{\omega}$ and $\omega_{\overline{A}}:= \iota_{\overline{A}}\circ \phi_1^* \boldsymbol{\omega}= \overline{f}\circ \phi_2^*\boldsymbol{\omega} $. 
\end{theorem}

An important distinction between Theorem \ref{THM: Integration formula} and the standard subadditivity property of Hausdorff integration is the inclusion of the additional boundary of the subspace $A$. Heuristically, we need to ensure that the restriction of $\boldsymbol \omega$ to $A$ is compactly supported, and the only way to do this is to include the boundary of $A$ within the integral. It may seem that the inclusion of this extra component is an innocuous prescription, given that boundaries are of measure zero. However, as we will see in Section 3, this extra boundary component will have some far reaching consequences for the non-Hausdorff Gauss-Bonnet theorems. 

%



\subsection{Metrics and Curvature}
We may construct metrics of arbitrary signature on a non-Hausdorff manifold by gluing together the spaces $M_1$ and $M_2$ along an isometry. In a global picture, we may view two metrics on the $M_i$ as sections of the appropriate tensor bundle for which the overlap condition of Theorem \ref{THM: sections are fibre product} manifests as an isometric equivalence on the gluing region $A$. It should be noted that there is no issue regarding the regularity of the resulting non-Hausdorff metric -- Theorem \ref{THM: sections are fibre product} ensures that any Lorentzian metric, viewed as a global section of the appropriate tensor bundle, is indeed smooth everywhere.  \\

Despite the non-existence of partitions of unity, affine connections may still be constructed in the non-Hausdorff setting. In global notation, an affine connection on $\textbf{M}$ is defined as per usual, that is, as a bilinear operator $$ \boldsymbol{\nabla}: \Gamma(T\textbf{M}) \times \Gamma(T\textbf{M}) \rightarrow \Gamma(T\textbf{M}), \ \boldsymbol{\nabla}(\textbf{v}, \textbf{w}) \mapsto \boldsymbol{\nabla}_{\textbf{v}} \textbf{w} $$
that is $C^\infty(\textbf{M})$-linear in the first argument, and satisfies the Leibniz rule: $ \nabla_\textbf{v}(\textbf{f}\textbf{w}) = \textbf{f}\nabla_\textbf{v}(\textbf{w}) + \mathcal{L}_\textbf{v}(\textbf{f}) \textbf{w} $ for all $\textbf{f} \in C^\infty(\textbf{M})$ and $\textbf{v}, \textbf{w} \in \Gamma(T\textbf{M})$.\footnote{The Lie derivative used here is defined as in the Hausdorff case, that is, via local flows of the vector fields. It can be shown that in the non-Hausdorff case, the Lie derivative is still a local operator that satisfies $\mathcal{L}_{\textbf{v}}\textbf{w} = [\textbf{v}, \textbf{w}]$ -- see Section 3.1 of \cite{o2023deRham} for a detailed discussion.} As with the non-Hausdorff sections of Section 2.2, it can be shown that any connection $\boldsymbol{\nabla}$ on $\textbf{M}$ will restrict to a pair affine connections $\nabla^i:=\phi_i^*\boldsymbol{\nabla}$ on $M_i$ that agree once mutually pulled back to $A$. As the following result states, the converse is also true: a pair of connections on $M_i$ may be ``glued" together to define a connection on $\textbf{M}$. 

\begin{lemma}[{\cite{o2023deRham}}]
    Suppose $\nabla_i$ are a pair of affine connections defined on the manifolds $M_i$. If $\iota_A^*\nabla_1 = f^*\nabla_2$ on $A$, then $\boldsymbol{\nabla}$ defined by 
    $$ (\boldsymbol{\nabla}_{\textbf{v}}\textbf{w})(\textbf{x}) = \begin{cases}
        \phi_1\left( (\nabla_1)_{(\phi_1^*\textbf{v})}(\phi_1^* \textbf{w})(\phi_1^{-1}(\textbf{x})) \right) & \text{if} \ \textbf{x} \in \phi_1(M_1) \subseteq \textbf{M} \\
        \phi_2\left( (\nabla_2)_{(\phi_2^*\textbf{v})}(\phi_2^* \textbf{w})(\phi_2^{-1}(\textbf{x})) \right) & \text{if} \ \textbf{x} \in \phi_2(M_2) \subseteq \textbf{M}
    \end{cases} $$ is an affine connection on $\textbf{M}$. 
\end{lemma}     
Observe that in the above, the assumption $\iota_A^*\nabla_1 = f^*\nabla_2$ is precisely what is needed to ensure that $\boldsymbol{\nabla}$ is well-defined, since by construction $\phi_1(M_1)\cap \phi_2(M_2) = \phi_1(A)\cap \phi_2(f(A))$. 
Using the above prescription, it can be shown that a Levi-Civita connection exists for any metric on $\textbf{M}$. Heuristically, although we don't have full access to partitions of unity for $\textbf{M}$, we \textit{do} have full access for the Hausdorff manifolds $M_i$. Therefore, we may construct pieces of the Levi-Civita connection on each $M_i$, and requiring that $f:A\rightarrow M_2$ be an isometric embedding ensures that these Hausdorff Levi-Civita connections may be transferred into $\textbf{M}$. Following on from this, we may define familiar geometric quantities such as the Riemann curvature tensor, the Ricci tensor and the Ricci scalar in the non-Hausdorff case, ultimately by appealing to the fact that $\textbf{M}$ is locally isometric to both $M_i$ \cite{o2023deRham}. \\

In a similar spirit, orientations of the manifolds $M_i$ may be glued in a manner consistent with Theorem \ref{THM: sections are fibre product} to yield an orientation of the non-Hausdorff manifold $\textbf{M}$. In such a situation, the canonical maps $\phi_i: M_i \rightarrow \textbf{M}$ become orientation-preserving isometries. \\

With an eye towards Section 4, we finish this section with a brief application of these ideas. Suppose for a moment that $\textbf{M}$ is a two-dimensional non-Hausdorff manifold with Riemannian metric $\textbf{h}$. The canonical maps $\phi_i:M_i \rightarrow \textbf{M}$ act as isometric embeddings, which means that $\textbf{M}$ may be locally isometric to either $M_i$, depending on where you are in the manifold. According to the integral formula of Theorem \ref{THM: Integration formula}, the total scalar curvature of $\textbf{M}$ may be written as 
\begin{equation}\label{EQ: NH scalar curvature in general}
    \int_{\textbf{M}} R \  dA  = \int_{M_1} R \ dA+ \int_{M_2} R \ dA - \int_{\overline{A}} R  \ dA 
\end{equation}
where here the metrics on the Hausdorff manifolds are the pullbacks of $\textbf{h}$ by the relevant maps, and each Ricci scalar and area form is computed via the pulled-back versions of $\textbf{h}$. 

\section{Gauss-Bonnet in Various Forms}
The Gauss-Bonnet theorem is a powerful result that relates the total scalar curvature of a two-dimensional manifold to its Euler characteristic. As explained in the introduction, our main strategy for evaluating the gravitational action of the non-Hausdorff trousers space will be via this particular theorem. As such, we will now spend some time discussing various versions of the Gauss-Bonnet theorem. In distinction to Section 2, throughout this section we will assume that all manifolds, Hausdorff or otherwise, are two-dimensional. \\

In Euclidean signature, the Gauss-Bonnet theorem for a manifold $(M,h)$ with boundary may be stated as the equality 
\begin{equation}\label{EQ: Euclidean GBT}
    2\pi \chi(M) = \frac{1}{2}\int_M R dA + \int_{\partial M} k d\gamma + \sum \theta_{ext},
\end{equation}  
where here $\chi(M)$ is the Euler characteristic of $M$,\footnote{Here we may define the Euler characteristic to be the alternating sum of the ranks of the simplicial homology groups of $M$, or equivalently as $\chi(M) = V_{\mathcal{T}} - E_{\mathcal{T}} + F_{\mathcal{T}}$ for any triangulation $\mathcal{T}$ of $M$.} and the boundary $\partial M$ is assumed to consist of finitely many piecewise-smooth connected components \cite{do2016differential}. The expressions $\theta_{ext}$ denote the exterior angles between adjacent smooth segments of a boundary component. Geometrically, these are the angles by which a vector must instantaneously turn at the non-smooth corners, as pictured in Figure \ref{FIG: turning angles}. We will often borrow from physics parlance and refer to these exterior angles as \textit{corner terms}. Given that the boundary $\partial M$ is a piecewise-smooth curve, the corner term $\theta_{ext}$ at a vertex $p$ is computed by taking the one-sided derivatives of the adjacent curves and computing the angle between the corresponding vectors lying in the tangent space $T_pM$.  \\

We will now set about modifying this theorem in two orthogonal directions: firstly, we will review the so-called \textit{Lorentzian Gauss-Bonnet theorem}, and then we will generalise everything to the non-Hausdorff setting. As mentioned in the introduction, there is a conceptual difficulty when proving the Lorentzian Gauss-Bonnet theorem, ultimately stemming from the ill-defined notion of angle within Minkowski space. So, before getting to any modifications of the Gauss-Bonnet theorem, we will first spend some time reviewing the literature on Lorentzian angles.  

\begin{figure}
    \centering
    \begin{tikzpicture}

    \fill[] (0,0) circle[radius=0.7mm] {};
    \fill[] (0,3) circle[radius=0.7mm] {};
    \fill[] (4,3) circle[radius=0.7mm] {};
    \fill[] (3,0) circle[radius=0.7mm] {};

    \draw[thick] plot[smooth, tension=0.6] coordinates{(0,0) (1,0.3) (3,0)}; 
    \draw[thick] plot[smooth, tension=0.6] coordinates{(3,0) (3.2, 1.5) (4,3)};
    \draw[thick] plot[smooth, tension=0.6] coordinates{(4,3) (3,3.5) (1,2.8) (0,3)};
    \draw[thick] plot[smooth, tension=0.6] coordinates{(0,3) (0,2) (-0.25,1.2) (0,0)};

    \node[] at (3.28,1.5) {\rotatebox[origin=c]{-20}{$\boldsymbol{\wedge}$ }};
    \node[] at (1.84,3.18) {\rotatebox[origin=c]{118}{$\boldsymbol{\wedge}$ } };
    \node[] at (1.5,0.18) {\rotatebox[origin=c]{-95}{$\boldsymbol{\wedge}$ }};
    \node[] at (-0.27,1.5) {\rotatebox[origin=c]{162}{$\boldsymbol{\wedge}$ }};

    \draw[dashed] (3,0)--(5,-0.35);
    \draw[dashed] (3,0)--(3,2.5);

    \draw[dashed] (4,3)--(5.2,5);
    \draw[dashed] (4,3)--(2.3,4.5);

    \draw[dashed] (0,3)--(-1.8,4);
    \draw[dashed] (0,3)--(0.2,1);
    
    \draw[dashed] (0,0)--(0.6,-2);
    \draw[dashed] (0,0)--(1.65,1);

    \draw[color=red, ->] (0.25,-0.8) to [out=20,in=285] (0.75,0.45);
    \draw[color=red, ->] (-0.65,3.35) to [out=250,in=180]  (0.08,2.25);
    \draw[color=red, ->] (4.1,-0.2) to [out=80,in=10] (3,0.9);
    \draw[color=red, ->] (4.45,3.8) to [out=150,in=45] (3.25,3.65);

    \node[] at (3.6,1.75) {$\gamma_1$};
    \node[] at (1.8,3.5) {$\gamma_2$};
    \node[] at (-0.6,1.5) {$\gamma_3$};
    \node[] at (1.8,-0.2) {$\gamma_4$};

    \node[] at (4,4.3) {$\theta_{12}$};
    \node[] at (-0.5,2.3) {$\theta_{23}$};
    \node[] at (0.8,-0.9) {$\theta_{34}$};
    \node[] at (4.4,0.5) {$\theta_{41}$};
    
\end{tikzpicture}
    \caption{Turning angles for a piecewise-smooth, oriented closed curve $\gamma$ that bounds a flat disk embedded in $\mathbb{R}^2$. Here the total geodesic curvature of $\gamma$ equals $ \int_\gamma k d\gamma = \sum_{i=1}^4 \int_{\gamma_i} k d\gamma + \sum \theta_{ext}. $ The exterior angles (red) are computed using tangent vectors at each marked vertex. Hopf's Umlaufsatz states that the sum of these exterior angles equals $2\pi$, and the Gauss-Bonnet theorem confirms that the Euler characteristic of the disk is equal to $1$.}
    \label{FIG: turning angles}
\end{figure}
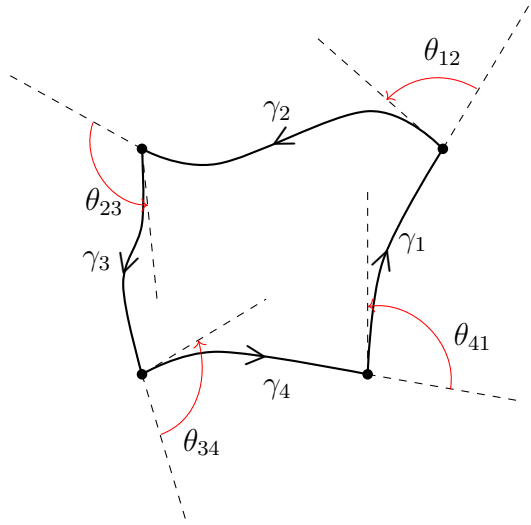

\subsection{Lorentzian Angles}

In a two-dimensional vector space with a fixed metric, the convex angle between a pair of normalised vectors may be determined from the parameter of the isometry transformation that sends one vector into the other. When working with the usual Euclidean metric, the isometry transformation lying in $SO(2)$ is an honest rotation of the unit circle. Up to a preferred orientation, the convex angle between two vectors in Euclidean space is uniquely determined -- ultimately because the action of $SO(2)$ on $\mathbb{R}^2$ is free and transitive. \\

This perspective also partially applies to Minkowski space. When passing to the $(-,+)$-signature of the Minkowski metric, the orbit spaces under $SO(1,1)$ of non-null vectors divide $\mathbb{R}^{1,1}$ into four disjoint quadrants. As such, a pair of Minkowski vectors $u$ and $v$ will admit an unambiguous and well-defined angle $\theta_{uv}$ provided that they are both non-null and are related by a boost. We will label the orbit space of the unit spacelike vector $v=(0,1)$ as Q1, and count the quadrants anticlockwise from there. Following \cite{sorkin2019lorentzian}, we introduce $$ Z(u,v) := u\cdot v +\sqrt{(u\cdot v)^2 - (u\cdot u)(v\cdot v)    } \ \  \text{and}  \ \ \overline{Z}(u,v) = u\cdot v -\sqrt{(u\cdot v)^2 - (u\cdot u)(v\cdot v)    }$$ as useful shorthands. We may then express the angle $\theta_{uv}$ between two spacelike vectors lying in the same quadrant as follows:
    \begin{align}\label{EQ: u,v both spacelike}
        \theta_{uv} & = \log{\frac{Z(u,v)}{|u| |v|}} & \text{if} \  u, v \ \text{both spacelike and in the same quadrant.}  
    \end{align}
This formula may be related to the familiar trigonometric expression of boosts by standard identities -- see \cite{asante2023complex} for the alternate form. Note that for spacelike vectors, the norm $|u| = \sqrt{u\cdot u}$ is real. \\

Since no boost can change the signature of a Minkowski vector, it may appear as though there is no meaningful notion of angle between vectors lying in different quadrants. The now-standard remedy for this issue is to analytically continue the meaningful fragments of angular formulae into the complex plane. With such a procedure, the result is a complex-valued notion of angle. Treatments of complex-valued Minkowski angles exist in various forms in the literature \cite{jia2022complex, asante2023complex, neiman2013imaginary, sorkin2019lorentzian}, though they typically differ in both scope and convention. For the purposes of this paper, we will opt for Sorkin's approach \cite{sorkin2019lorentzian}, since his treatment is sufficiently general so as to include both null vectors and a Gauss-Bonnet theorem. \\

As an illustration, we will now outline the derivation of a complex angle between two Lorentzian vectors $a:= (0,1)$ and $b:=(1,0)$. Using the null basis $m := (\frac{1}{2},-\frac{1}{2})$ and $n:=(\frac{1}{2},\frac{1}{2})$, we may write $a = n-m$ and $b = n+m$. An interpolating vector $c$ lying in the convex wedge between $a$ and $b$ may be described as $ c = m + \lambda n,  $ where $\lambda \in [-1,1]$. Allowing $\lambda$ to smoothly vary from $-1$ to $1$ will trace out a continuous transformation of $a$ into $b$. As we do this, we see that the expression $\frac{Z(a,c)}{|a| |c|}$ will become singular as $c$ becomes null when crossing quadrants at $\lambda=0$. This pole may be avoided by endowing the Minkowski metric with a small positive-definite imaginary part (cf. \cite{louko1997complex}), which adjusts the dot product by $ c\cdot c \rightarrow c\cdot c \pm     i\epsilon$ for vectors $c$ with norm close to zero. This allows us to circumvent the singularity at $\lambda=0$ and continue into the adjacent quadrant without issue.  \\

\begin{figure}
    \centering
\begin{tikzpicture}[scale=0.6]

\draw[dashed] (-5.5,-5.5)--(5.5,5.5);
\draw[dashed] (-5.5,5.5)--(5.5,-5.5);
\draw[dotted] (4,0)--(0,4);

\draw[red, thick, ->] (0,0)--(4,0);
\draw[red, thick, ->] (0,0)--(0,4);
\draw[thick, ->] (0,0)--(2,2);
\draw[thick, ->] (0,0)--(-2,2);
\draw[thick, blue, ->] (0,0)--(3,1);

\node[] at (-2,1.3) {$m$};
\node[] at (1.3,1.8) {$n$};
\node[] at (2.5,-0.35) {$a$};
\node[] at (-0.35,2.5) {$b$};
\node[] at (1.9,1) {$c$};

\node[] at (6,0) {\textsf{Q1}};
\node[] at (0,-6) {\textsf{Q4}};
\node[] at (-6,-0) {\textsf{Q3}};
\node[] at (0,6) {\textsf{Q2}};

\end{tikzpicture}
    \caption{Sorkin's derivation of the Lorentzian angle between a spacelike and a timelike vector. Here we pick the lightcone basis $m:=(\frac{1}{2},-\frac{1}{2})$ and $n=(\frac{1}{2},\frac{1}{2})$, for which $a = n-m$ and $b=n+m$, and $c = n+\lambda m$. We may then continuously trace from $a$ to $b$, using a parameter $\lambda\in[-1,1]$. This parameter will yield a singularity in $\frac{Z(a,c)}{|a| \, |c|}$ at $\lambda=0$, which we may avoid via an $i\epsilon$-regularisation.}
    \label{FIG: Lorentzian angles}
\end{figure}
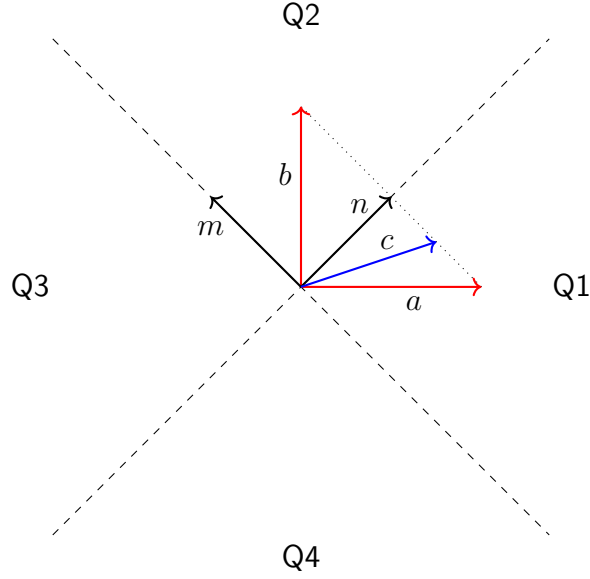

There are two subtleties to consider here. Firstly, we are trying to complexify the ratio $\frac{Z(a,c)}{|a| |c|}$, which means that we need to select a branch of the complex logarithm that will eventually be used as in \eqref{EQ: u,v both spacelike}. Following \cite{sorkin2019lorentzian}, we select the principle branch of the logarithm, so that $log(\frac{1}{i})= \frac{-\pi i}{2}$. Secondly, when performing the circumvention of the pole at $\lambda=0$, the sign of the $i\epsilon$-regulariser will dictate the sign of the (imaginary) norm for timelike vectors: $|u| := \sqrt{u\cdot u} = \pm i \sqrt{|u\cdot u|}$. With this in mind, the Lorentzian angle between our chosen vectors $a$ and $b$, will be purely imaginary:
$$    \theta_{ab}  =  \log\left({\frac{Z(a,b)}{|a| \  |b|}}\right)  = \log\left(\frac{1}{\pm i}\right) = \mp \frac{i\pi}{2},$$
the sign of which depends on the choice of $\pm i\epsilon$. For a general spacelike vector $u$ in Quadrant $1$ and a timelike vector $v$ lying in Quadrant $2$, there may also be a real part of the Lorentzian angle. The general angular formula will be
\begin{equation}\label{EQ: u spacelike, v timelike}
    \theta_{uv} =\log\left({\frac{Z(u,v)}{|u| \ |v|}}\right) =  \log\left({\frac{Z(u,v)}{||u|| \  ||v||}}\right) \mp \frac{\pi i}{2},
\end{equation}
where here in the latter term we isolate the imaginary component of the angle by taking $|| \cdot ||$,  the absolute value of the norm $| \cdot |$. \\

In similar spirit, the Lorentzian angle between a pair of timelike vectors lying in the same quadrant may be determined once we have fixed a sign for the complex quantity $|u| = \sqrt{u \cdot u}$. In contrast to \eqref{EQ: u,v both spacelike}, we have: 
\begin{align}\label{EQ: u,v both timelike}
        \theta_{uv} & = - \log{\frac{\overline{Z}(u,v)}{|u| |v|}} &  \text{if} \  u, v \ \text{both timelike and in the same quadrant.}  
    \end{align}
Using the above, one can interpret the imaginary contribution of $\mp \frac{\pi i}{2}$ in the formula \eqref{EQ: u spacelike, v timelike} as a discrete ``rotation" of the spacelike vector into Quadrant $2$, followed by an application of formula \eqref{EQ: u,v both timelike} to determine the remaining real part of the angle. \\

Generally speaking, the $i\epsilon$-regularisation causes a discrete contribution of $\mp\frac{i\pi}{2}$ whenever we compute angles between non-null vectors in adjacent quadrants of Minkowski space. One can imagine a Euclidean rotation in which we trace one vector through null rays into adjacent quadrants (cf. Figure \ref{FIG: Lorentzian angles}). Every time a vector passes through a null ray it receives a discontinuous contribution of $\mp\frac{\pi i}{2}$. The total angle around the origin will then equal $\mp 2\pi i$, as depicted in \cite{neiman2013imaginary, asante2023complex}. \\

The $i\epsilon$-regularisation of \cite{louko1997complex, sorkin2019lorentzian, asante2023complex, neiman2013imaginary, jia2022complex} may also be used to define angles between null vectors. Combining equation \eqref{EQ: u spacelike, v timelike} with the additivity of angles, and proceeding via a case distinction, the angle involving null vectors may take any of the following forms: 
\begin{align}
    \theta_{m,u} &= \log \left( \frac{m \cdot u}{ l_0||u|| } \right) \mp \frac{i\pi }{4} & \text{if} \ u \ \text{spacelike} \label{EQ: m null, u spacelike} \\
    \theta_{m,v} &= \log \left( \frac{m\cdot v}{ l_0||v|| } \right) \mp \frac{i\pi }{4} & \text{if} \ v \ \text{timelike} \label{EQ: m null, u timelike}\\
    \theta_{m,n} &= \log \left( \frac{m \cdot n}{ l_0^2 } \right) \mp \frac{i\pi }{2} & \text{if} \ n \ \text{null and} \  m,n \ \text{bound spacelike quadrant} \label{EQ: m,n null, spacelike quadrant} \\
    \theta_{m,n} &= -\log \left( \frac{m \cdot n}{ l_0^2 } \right) \mp \frac{i\pi }{2} & \text{if} \ n \ \text{null and} \  m,n \ \text{bound timelike quadrant} \label{EQ: m,n null, timelike quadrant}
\end{align}
Here we follow \cite{sorkin2019lorentzian} and introduced two more conventions. Firstly, due to the additivity of angles, we need to choose how to divvy up the imaginary contribution of $\mp \frac{\pi i}{2}$ arising from \eqref{EQ: u spacelike, v timelike} into the two formulae (7) and (8). In the above we have opted for a balanced contribution of $\mp \frac{\pi i}{4}$ for either side of the null vector. Secondly, we commit to an additional variable $l_0$, which is an arbitrary length scale that is required in order to make the angle formulae dimensionless (cf. the $4d$ corner ambiguites in \cite{jubb2017boundary, lehner2016gravitational}). Although both of these conventions are choices, they will ultimately not affect the Gauss-Bonnet theorem once we keep them consistently fixed.

\subsection{A Gauss-Bonnet Theorem for Surfaces with Null Boundaries}

A Lorentzian version of the Gauss-Bonnet theorem was originally proved implicitly by Chern in \cite{chern1963pseudo}, who extended his famous Chern-Gauss-Bonnet theorem to closed, even-dimensional Lorentzian manifolds. Various alternate versions of the result exist in the literature, and each varies in scope and generality. Older resources such as \cite{avez1963formule, alty1995generalized, law1992neutral, jee1984gauss} prove the theorem by considering manifolds with either empty or non-null boundary. This was generally due to an underdeveloped notion of of Lorentzian angles at their time of writing. In recent years, however, Sorkin \cite{sorkin2019lorentzian} extended the Lorentzian Gauss-Bonnet theorem to include null boundary components. We will now briefly overview his argument, and flesh out some of the more rigorous details. \\

We start by proving a Gauss-Bonnet theorem for the Lorentzian triangle. This version, commonly called the ``local Gauss-Bonnet theorem" in Euclidean terminology, ultimately follows from an analogue of Hopf's \textit{Umlaufsatz} -- cf. the Euclidean version in Figure \ref{FIG: turning angles}. In this context, the Umlaufsatz states that the sum of turning angles around the oriented boundary of a Lorentzian triangle will always equal $\mp 2\pi i$. Heuristically speaking, we can imagine parallel transporting a vector around the boundary of the triangle and summing up the discrete jumps at the corners. Since the triangle is assumed to be flat, the vector will return to itself with no real angular defect, however it will have collected a full rotation of Minkowski angles along its journey. \\

Formally, the Lorentzian Umlaufsatz may be proved via a case distinction on all the different types of triangles. For triangles with non-null edges, the Umlaufsatz was discussed by both Jee \cite{jee1984gauss} and Law \cite{law1992neutral}. For triangles involving null edges, the discussion was continued in \cite{sorkin2019lorentzian}. For the purposes of illustration, we consider the case of a triangle $\Delta$ with two null edges $m$ and $n$, and a spacelike edge $w$. After orienting $\partial \Delta$, we may consider the turning angles by fixing null vectors with an affine length equal to that of the corresponding edge of $\Delta$. Using the additivity of Lorentzian angles, we can represent the sum of turning angles as a sum between purely null edges: 
$$  \theta_{mn} + \theta_{nw} + \theta_{wm} = \theta_{mn} + (\theta_{n(-m)} + \theta_{(-m)w}) + (\theta_{w(-n)} + \theta_{(-n)m}) = \theta_{mn} + \theta_{n(-m)} + \theta_{(-m)(-n)} + \theta_{(-n)m} .        $$
Heuristically, we may imagine parallel transporting all three vectors to the same point, and then using additivity of angles to cancel out the terms containing the spacelike vector $w$. This intuition is depicted in Figure \ref{fig:Lorentzian umlaufsatz}. From this perspective, the Umlaufsatz for $\Delta$ is clear, provided that the ambiguous length scale $l_0$ is kept fixed throughout the sum of the turning angles. We should also note that this version of the Umlaufsatz is equivalent to Sorkin's observation that the interior angles of a Minkowskian triangle will sum up to the flat half-value $\mp i\pi$, since: 
    $$ \mp 2\pi i =  \sum_{i=1}^3 \theta_{ext} = \sum_{i=1}^3 (\mp i\pi - \theta_{int}) = \mp 3\pi i - \sum_{i=1}^3 \theta_{int}.   $$

Using this result for triangles in Minkowski space, we may then triangulate a given almost-Lorentzian surface, making sure to arrange any causal irregularities onto vertices of the triangulation, and then derive a discrete version of the Gauss-Bonnet by the same arguments found in the Euclidean case -- for instance those of \cite[Sec 4.5]{do2016differential}. Provided that we pick the same sign of the $i\epsilon$-regulariser for all the vertices in a given triangulation,\footnote{This assumption is implicit within \cite{sorkin2019lorentzian}, though is explicitly called the ``global Wick rotation" in \cite{asante2023complex}.} we can then express curvature via a Regge action which sums over Lorentzian defect angles. This yields the following version of the Gauss-Bonnet theorem, amenable to two-dimensional Regge calculus in Lorentzian signature.   \\
\begin{figure}
    \centering
\begin{tikzpicture}
    
    \fill[] (4,0) circle (0.7mm);
    \fill[] (-1,-1) circle (0.7mm);
    \fill[] (2,2) circle (0.7mm);

    \draw[->] (-1,-1)--(3.9,0);
    \draw[->] (4,0)--(2.05,1.95);
    \draw[->] (2,2)--(-0.93,-0.93);

    \draw[dotted] (4,0)--(6.5,0.5);
    \draw[dotted] (2,2)--(0.5,3.5);
    \draw[dotted] (0,0)--(-2.5,-2.5);

    \node[] at (1.7,-0.8) {$w$};
    \node[] at (0.5,0.8) {$n$};
    \node[] at (3.25,1.2) {$m$};

    \node[] at (0.85,2) {$\theta_{mn}$};
    \node[] at (-0.5,-1.7) {$\theta_{nw}$};
    \node[] at (4.5,0.9) {$\theta_{wm}$};

    \draw[->] (1.5,2.5) to [out=200,in=155] (1.5,1.5);
    \draw[->] (-1.45,-1.45) to [out=340,in=270] (-0.35,-0.87);
    \draw[->] (4.8,0.15) to [out=100,in=60] (3.5,0.5);

    
    \fill[] (10,1) circle (0.7mm); 
    \draw[->] (10,1)--(15,2);
    \draw[->] (10,1)--(7,-2);
    \draw[->] (10,1)--(8,3);
    \draw[->, color=gray!65] (10,1)--(13,4);
    \draw[->, color=gray!65] (10,1)--(12,-1);

    \node[] at (8.5,1) {$\theta_{mn}$};
    \node[] at (10,2.35) {$\theta_{(-n)m}$};
    \node[] at (10,-0.35) {$\theta_{n(-m)}$};
    \node[] at (11.5,1.65) {$\theta_{w(-n)}$};
    \node[] at (11.65,0.5) {$\theta_{(-m)w}$};

    \draw[->] (10.75,1.75) to [out=135,in=45] (9.25,1.75);
    \draw[->] (9.25,1.75) to [out=215,in=135] (9.25,0.25);
    \draw[->] (9.25,0.25) to [out=315,in=215] (10.75,0.25);
    \draw[->] (10.75,0.25) to [out=45,in=315] (10.75,1.75);

    \node[] at (8,2.5) {$m$};
    \node[] at (7.3,-1.4) {$n$};
    \node[] at (14.5,2.2) {$w$};
    \node[color=gray!65] at (11.35,-1) {$-m$};
    \node[color=gray!65] at (12.3,3.8) {$-n$};

\end{tikzpicture}
    \caption{The intuition behind the Umlaufsatz of a Lorentzian triangle with two null edges. Here we use the point-vector correspondence of flat space to transport all vectors to the same point, and observe the manifest value of $\mp 2\pi i$.}
    \label{fig:Lorentzian umlaufsatz}
\end{figure}
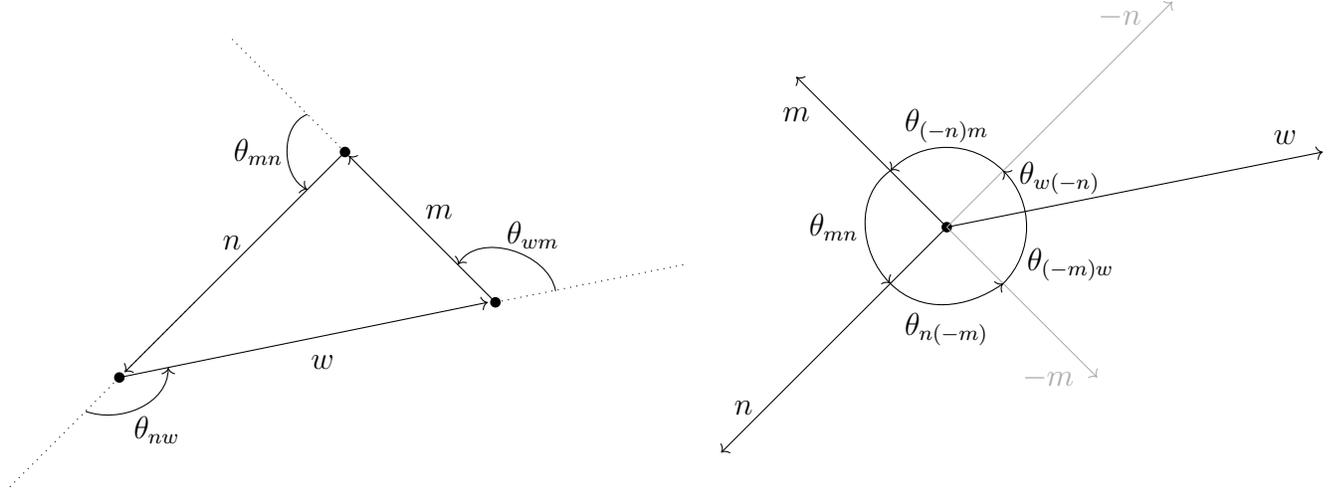
\begin{theorem}[Adapted from {\cite{sorkin2019lorentzian}}]\label{THM: Combinatorial LGBT}
        Let $M$ be an almost-Lorentzian surface with boundary, and let $\mathcal{T}$ be a triangulation of $M$ in which the degenerate points of $M$ lie on vertices of $\mathcal{T}$. Denote by $V_M$ and $V_{\partial M}$ the set of vertices of $\mathcal{T}$ that lie in the bulk and boundary of $M$, respectively. Then $$ \mp 2\pi i \chi(M) = \sum_{p \in V_M} \delta^2(p) + \sum_{q\in V_{\partial M}} \delta^1(q), $$
        where $\delta^2(p) = \mp 2\pi i - \sum \theta(p)$ and $\delta^1(q) = \mp \pi i - \sum \theta(q)$ measure the defect angles in $2d$ and $1d$, respectively.
\end{theorem}   

We may derive a continuum version of the above by measuring curvature by different means. In the triangulation argument of Regge calculus we flatten the triangles, thereby localising any curvature to vertices. Alternatively, we could also determine the curvature of a triangle $\Delta$ by measuring the defect in the sum of its interior angles. Following \cite{cheeger1984curvature}, by taking a limit of smaller and smaller triangles around a point, we may define an infinitesimal notion of curvature which may then be integrated to determine the total scalar curvature of the bulk. Similarly, for each smooth boundary component, we may determine its geodesic curvature via a limit of defect angles of small lines surrounding points. Thus we obtain the following (smooth) version of the Lorentzian Gauss-Bonnet theorem. 

    \begin{theorem}\label{THM: smooth LGBT}
        Let $M$ be an almost-Lorentzian manifold with boundary. Then $$ \frac{1}{2}\int_M R dA + \int_{\partial M} k d\gamma + \sum \theta_{ext}= \mp 2\pi i \chi(M).   $$
    \end{theorem}

\subsection{Non-Hausdorff Gauss-Bonnet Theorems}

We can repeat the philosophy of Section 2 in order to prove a Gauss-Bonnet theorem for non-Hausdorff manifolds. Our approach will be to first relate the Euler characteristic of a non-Hausdorff manifold to the Euler characteristics of its Hausdorff submanifolds. After this, we may apply either the Euclidean or Lorentzian Gauss-Bonnet Theorem to all of these Hausdorff pieces, and then collate the resulting integrals into a global non-Hausdorff curvature term using the subadditivity formula of Theorem \ref{THM: Integration formula}. \\
 
In our discussion thus far we have assumed some simplicial definitions of the Euler characteristic. However, since all simplicial complexes are by construction Hausdorff topological spaces, we do not have access to simplicial homology for non-Hausdorff manifolds. There are several ways to bypass this issue. It is reasonable to suggest that there may be a non-Hausdorff version of simplicial homology, with a well-defined boundary operator that may in turn yield some kind of topological invariants. However, we will avoid this line of inquiry and instead appeal to a broader definition of Euler characteristic that holds for more-general topological spaces. We define the Euler characteristic of a non-Hausdorff $n$-dimensional manifold $\textbf{M}$ as follows: 
    $$\chi(\textbf{M}) = \sum_{i=1}^n(-1)^i\textrm{rank}(H^{S}_i(\textbf{M})), $$
where here we use the singular homology groups $H^S_i (\textbf{M})$. In the Hausdorff setting, this definition of Euler characteristic coincides with the ordinary definition via triangulations, due to the equivalence between simplicial and singular homology \cite{hatcher2002algebraic}. In the case of smooth non-Hausdorff manifolds, singular cohomology is known to be isomorphic to de Rham cohomology for any manifold satisfying the conditions of Theorem \ref{THM: smooth nh manifold} \cite{o2023deRham}. \\
    
Under this reading, we may now relate the Euler characteristic of a non-Hausdorff manifold to those of its Hausdorff constituents. The following result confirms that the familiar ``inclusion-exclusion principle" (cf. pg. 221 of \cite{rotman2013introduction}) holds for our non-Hausdorff manifolds.
    \begin{theorem}[{\cite{o2023deRham}}]\label{THM: nH Euler characteristic}
        Let $\textbf{M}$ be a non-Hausdorff manifold, defined according to Theorem \ref{THM: topological nh manifold}. Then $$ \chi(\textbf{M}) = \chi(M_1) + \chi(M_2) - \chi(A).  $$ 
    \end{theorem}
\begin{proof}
    We provide a sketch of the proof; for details, see \cite{o2023deRham}. Firstly, we may observe that the singular homology groups of $\textbf{M}$ may be related to the homology groups of the Hausdorff submanifolds by the Mayer-Vietoris sequence: 
    \begin{center}
\adjustbox{scale=0.85}{\begin{tikzcd}[column sep= 3.2em, row sep=2.5em, arrow style=math font,cells={nodes={text height=2ex,text depth=0.75ex}}]
      
       0 \arrow[r] & \cdots \arrow[r] & H^S_{q+1} (A) \arrow[r, "i_* \oplus f_*"] & H^S_{q+1}(M_1) \oplus  H^S_{q+1}(M_2) \arrow[r,"\phi_{1*} - \phi_{2*}"] \arrow[draw=none]{d}[name=Y, shape=coordinate]{} & H^S_{q+1}(\textbf{M}) \arrow[curarrow=Y]{dll}{} & \\
       & & H^S_{q} (A)\arrow[r, "i_* \oplus f_*"] & H^S_{q}(M_1) \oplus  H^S_{q}(M_2)\arrow[r, "\phi_{1*} - \phi_{2*}"]  & H^S_{q}(\textbf{M}) \arrow[r] & \cdots \arrow[r] & 0
\end{tikzcd}}
\end{center}
The precise description of these maps is not too important for our purposes, but can be found in \cite{o2023deRham} (they are essentially a reformulation of the pushforward maps found in \cite{hatcher2002algebraic}[Sec. 2.2], descended to homology). The important observation is that the above is a long exact sequence of vector spaces that terminates. Crucially, for any long exact sequence of vector spaces, the alternating sum of the dimensions of all the entries of the sequence must be zero. The result then follows by a rearrangement of this vanishing alternating sum. 
    \end{proof} 

In order to prove a non-Hausdorff Gauss-Bonnet theorem, we need to make a subtle yet crucial observation. In the integral formula of Theorem \ref{THM: Integration formula}, we must integrate over the topological closure of $A$, which in this context means that we are integrating a manifold with boundary. However, in the subadditive formula of $\chi(\textbf{M})$ listed above, we instead have an expression for $\chi(A)$, and in general it is not the case that $\chi(A) = \chi(\overline{A})$.\footnote{In the case that the $M_i$ are manifolds without boundaries, this tension may be broken by assuming that the open subset $A$ is \textit{regular-open}, meaning that the interior of $\overline{A}$ equals $A$ itself. This allows a homotopy equivalence between $A$ and $\overline{A}$, which induces an isomorphism of singular homologies and guarantees the equality $\chi(A) = \chi(\overline{A})$.} However, when we work with the non-Hausdorff trousers space in Section 4 we will not need to worry about this potential issue, so in what follows we will simply assume the equality $\chi(A)=\chi(\overline{A})$. \\

According to the remark we made at the end of Section 2.4, we may apply the subadditivity principle of Theorem \ref{THM: Integration formula} to a $2d$ curvature form to yield the equality \eqref{EQ: NH scalar curvature in general}. If we apply the Gauss-Bonnet theorem to each Hausdorff component of $\textbf{M}$, we will obtain the equality:  
\begin{align*}
\frac{1}{2}\int_{\textbf{M}} R \  dA &  = \frac{1}{2}\int_{M_1} R \ dA + \frac{1}{2}\int_{M_2} R \ dA - \frac{1}{2}\int_{\overline{A}} R  \ dA\\
     & = \left( 2\pi \chi(M_1) -\int_{\partial M_{1}} k d\gamma \right) + \left(2\pi \chi(M_2) -\int_{\partial M_{2}} k d\gamma\right) - \left(2\pi \chi(\overline{A}) - \int_{\partial \overline{A}} k dx \right) \\
     & = 2\pi \chi(\textbf{M}) - \left(  \int_{\partial M_{1}} k d\gamma + \int_{\partial M_{2}} k d\gamma - \int_{\partial \overline{A}} k d\gamma \right).
\end{align*}
Here we assume that the boundary of $\textbf{M}$ consists of the images of $\partial M_i$ under the canonical maps $\phi_i$. Note that it's possible for the two sets $\phi_i(\partial M_i)$ to intersect, however this may only occur if $M_1$ and $M_2$ have a common manifold boundary along the gluing region $A$. Thus, the closure $\overline{A}$ may have two different types of boundary: a \textit{manifold} boundary that already exists within $A$, and a \textit{topological} boundary that forms the Hausdorff-violating points once mapped into $\textbf{M}$. We thus write $\partial \overline{A} = \partial A \sqcup \textsf{Y}$, where here we use the special symbol $\textsf{Y}$ to denote the Hausdorff-violating piece of $\overline{A}$. By the usual subadditivity of Hausdorff integration, we see that
$$  \int_{\partial \textbf{M}} k  d\gamma = \int_{\partial M_1} k  d\gamma + \int_{\partial M_2} k  d\gamma - \int_{\partial A} k  d\gamma \quad \text{and} \quad \int_{\partial \overline{A}} k d\gamma = \int_{\partial A} k d\gamma + \int_{\textsf{Y}} k d\gamma. $$
By substituting these equalities into the previous equation, we obtain a non-Hausdorff Gauss-Bonnet in Euclidean signature:
\begin{equation}\label{EQ: Euclidean nh GBT}
    \frac{1}{2}\int_{\textbf{M}} R \  dA =  2\pi \chi(\textbf{M}) - \left(  \int_{\partial\textbf{M}} k d\gamma  - \int_{\textsf{Y}} k  d\gamma \right).
\end{equation}
We may the apply essentially the same reasoning as the above, this time using the Lorentzian Gauss-Bonnet theorem of Theorem \ref{THM: smooth LGBT}, to deduce the following. 

\begin{theorem}\label{THM: Lorentzian non-Hausdorff Gauss-Bonnet}
Let $(\textbf{M}, \textbf{g})$ be a non-Hausdorff two-dimensional spacetime built from Hausdorff spacetimes $M_1$ and $M_2$ according to Theorem \ref{THM: smooth nh manifold}. Suppose furthermore that the $M_i$ are manifolds with boundary, and $A$ satisfies $\chi(A)=\chi(\overline{A})$. Then $$ \mp 2\pi i\chi(\textbf{M}) = \frac{1}{2}\int_{\textbf{M}} R dA + \int_{\partial \textbf{M}} k d\gamma - \int_{\textsf{Y}} k d\gamma. $$
\end{theorem}

\section{The Non-Hausdorff Trousers Space}
We will now define and evaluate the gravitational action for a non-Hausdorff version of the Trousers space. As outlined in the introduction, we will consider a compactified version of Penrose's spacetime of Figure \ref{FIG: Penrose spacetime}, taken to be long enough so that the initial surface $\Sigma_1$ is homeomorphic to $S^1$ and the final surface $\Sigma_2$ is homeomorphic to $S^1 \sqcup S^1$. According to our discussion in Section 2, we can construct such a space by gluing two copies of the cylinder together everywhere outside the causal future of a point. Formally, we take
\begin{itemize}
    \item $M_1 = M_2$ to be cylinders $S^1\times D^1$ endowed with the same flat Lorentzian metric,
    \item a preferred point $p$ in $M_1$ whose causal future $J^+(p)$ contains the final boundary $S^1\times \{1\}$ of $M_1$ as a subset,  
    \item $A = M_1 \backslash J^+(p)$ endowed with the open submanifold atlas and metric induced from $M_1$, and
    \item $f: A\rightarrow M_2$ to be the identity map. 
\end{itemize}
We denote by $\textbf{T}$ the topological space formed according to the above. Since the gluing map $f$ is taken to be the identity map, the data above falls within the scope of Theorems \ref{THM: topological nh manifold}, \ref{THM: smooth nh manifold} and \ref{THM: sections are fibre product}, and thus we may view $\textbf{T}$ as a smooth non-Hausdorff manifold that contains the two cylinders $M_i$ as maximal Hausdorff open submanifolds. Moreover, the map $f$ is also a time-orientation preserving isometric embedding, so we may effectively transfer the Lorentzian metrics of the cylinders $M_i$ into the non-Hausdorff manifold $\textbf{T}$ according to our discussion in Section 2.4. We denote the resulting non-Hausdorff spacetime by $(\textbf{T}, \textbf{g)}$, and we will hereafter refer to this as the \textit{non-Hausdorff trousers space}. \\

By construction, the $M_1$-relative closure $\overline{A}$ of the gluing region $A$ will include the lightlike future of the point $p$, as depicted in Figure \ref{FIG: A-bar}. Once the quotient is performed to construct $(\textbf{T}, \textbf{g)}$, the Hausdorff-violating submanifold of $\textbf{T}$ will be the two lightlike futures of the distinct points $\phi_1(p)$ and $\phi_2(p)$. The height function $\textbf{f}:\textbf{T} \rightarrow \mathbb{R}$ that maps each point in $\textbf{T}$ to its $D^1$-coordinate in either $M_i$ will be a well-defined smooth function that allows us to interpret $(\textbf{T}, \textbf{g})$ as a non-Hausdorff transition from $S^1$ to $S^1 \sqcup S^1$. Indeed -- on the spacelike slice of $\textbf{T}$ that contains the points $\phi_i(p)$, the topology will begin splitting by changing from $\Sigma_1=S^1$ to a circle with two Hausdorff-violating points. These points will then propagate along null rays, and finish splitting into two circles at some later time. At every point in time after this, the spacelike slices of $(\textbf{T}, \textbf{g})$ will have the topology of $S^1 \sqcup S^1$. 

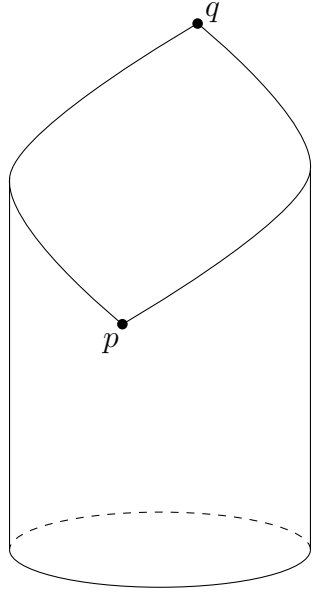
\begin{figure}
    \centering
    
\begin{tikzpicture}
    \draw (2,0) arc
	[start angle=360,	end angle=180, x radius=20mm, y radius =5mm] ;
\draw[dashed] (2,0) arc
	[start angle=0,	end angle=180, x radius=20mm, y radius =5mm] ;

 \draw[] (-2,0)--(-2,4.95);
 \draw[] (2,0)--(2,5.05);
\fill[] (-0.5,3) circle (0.7mm);
\fill[] (0.5,7) circle (0.7mm);

\draw[] plot[smooth, tension=0.6] coordinates{(-0.5,3) (1.99,5) (0.5,7)};
\draw[] plot[smooth, tension=0.6] coordinates{(-0.5,3) (-1.99,5) (0.5,7)};

\node[] at (-0.65,2.75) {$p$};
\node[] at (0.7,7.15) {$q$};
 
\end{tikzpicture}\caption{The region $\overline{A}$ for the non-Hausdorff trousers space $\textbf{T}$. When constructing $\textbf{T}$ we glue two cylinders together everywhere in $A$, yet leave the future component of $\partial A$ unidentified. In the quotient space $\textbf{T}$, this yields two Hausdorff-inseparable copies of the lightlike future of the point $p$.}
    \label{FIG: A-bar}
\end{figure}

\subsection{Causal Properties}
As mentioned in the introduction, one may justify the inclusion of the Hausdorff trousers space into a path integral scheme such as Equation \ref{EQ: schematic Lorentzian path integral} on causal grounds. Despite being non-globally hyperbolic, the Hausdorff trousers space, aside from the causal irregularity at the crotch, is a well-behaved spacetime. In particular, it exhibits no closed timelike curves, admits a global time function, and determines a causal structure whose relation $\leq$ is a partial order. \\

We will now show that the non-Hausdorff trousers space $(\textbf{T}, \textbf{g})$ enjoys similar, and even superior, causal features. To begin, we recall the relation $\textbf{x} \leq \textbf{y}$ if and only if $\textbf{y}$ lies in the future lightcone of $\textbf{x}$, which in this context means that there is a future-directed causal curve connecting $\textbf{x}$ to $\textbf{y}$. The following result characterises the causal relation $\leq$ of $\textbf{T}$ in terms of the Hausdorff spacetimes $M_i$.

\begin{lemma}\label{LEM: order theory of nh trousers space}
    For any two points $\textbf{x}$ and $\textbf{y}$ in $\textbf{T}$, we have that $\textbf{x} \leq \textbf{y}$ if and only if  $\phi_1^{-1}(\textbf{x}) \leq \phi_1^{-1}(\textbf{y})$ or $ \phi_2^{-1}(\textbf{x})\leq \phi_2^{-1}(\textbf{y})$, or both.  
\end{lemma} 

In this sense, we can see that the non-Hausdorff Trousers space $\textbf{T}$ naturally inherits the causal structures of $M_1$ and $M_2$. In particular, we see that the causal structure $(\textbf{T}, \leq)$ naturally inherits a poset structure from the causality of the cylinders $M_i$. \footnote{In fact, it appears that a more precise categorisation of the features of the causal space $(\textbf{T}, \leq)$ is that it manifests as a model of the logico-mathematical theory $\textsf{BST}_{NF}$ of \cite{belnap2021branching, OConnellthesis, wronski2009minkowskian}. This is an order-theoretic axiomatisation of a certain type of indeterminism, and runs somewhat parallel to the considerations of the causal set theorists.}

\begin{theorem}
    The causal relation $\leq$ on $\textbf{T}$ is a transitive, reflexive and antisymmetric relation.
\end{theorem}

As a brief aside, we remark that the above appears to hold more generally -- if we consider some category of non-necessarily-Hausdorff spacetimes with isometric embeddings as morphisms, then the mapping of spacetimes to causal sets may be considered as a covariant functor into the category of partially-ordered sets. In this case, it is likely that the above argument generalises directly to conclude that this ``causal functor" takes colimits to colimits, so that we can describe any causal structure on a non-Hausdorff spacetime as a colimit of posets. \\

Aside from the poset structure of $(\textbf{T}, \leq)$, the non-Hausdorff trousers space satisfies other nice causal properties, which do \emph{not} hold in the Hausdorff trousers case. First, recall from pointset topology that the compactness of a set is preserved under continuous maps. Applied to our context, we may use the canonical maps $\phi_i: M_i \rightarrow \textbf{T}$ to send causal diamonds $J^+(x) \cap J^-(y)$ from either $M_i$ into $\textbf{T}$. Since the Lorentzian cylinder is globally hyperbolic, Lemma \ref{LEM: order theory of nh trousers space} allows us to conclude that the causal diamonds $J^+(\textbf{x}) \cap J^-(\textbf{y})$ are compact in $\textbf{T}$. Moreover, one can readily find spacelike slices of $\textbf{T}$ whose domain of dependence equals the whole space. In these two senses, the non-Hausdorff trousers space is causally better-behaved than the Hausdorff trousers. In fact, these properties are strongly reminiscent of global hyperbolicity -- a connection that should be more fully explored by studying field equations on the non-Hausdorff spacetime.

\subsection{The Gravitational Action}
We will continue our discussion in a similar spirit to Section 3.3, that is, we will first describe a non-Hausdorff action in Euclidean signature, and then we will move on to the subtleties of Lorentzian signature. In each case we will begin with a general treatment of non-Hausdorff manifolds, before evaluating the derived actions for the case of the non-Hausdorff manifold $\textbf{T}$, equipped with the metric of appropriate signature. As such, throughout this section we will we assume that $\textbf{M}$ is a non-Hausdorff two-dimensional manifold satisfying the topological assumptions of Theorem \ref{THM: Lorentzian non-Hausdorff Gauss-Bonnet} regarding the Euler characteristics of $A$ and $\overline{A}$. 

\subsubsection{Euclidean Gravity}
For a two dimensional Hausdorff manifold $M$ with possible boundary $\partial M$, vacuum Euclidean gravity may be fully described via the Gauss-Bonnet action:
$$  \mathcal{S}(M,h) = \frac{1}{2\kappa}\int_M R \ dA + \frac{1}{\kappa}\int_{\partial M} k d\gamma, $$
where here any potential corner contributions have been absorbed into the Gibbons-Hawking-York boundary term. According to the Gauss-Bonnet theorem for Riemannian surfaces, we know that $\mathcal{S}(M,h) = \frac{2\pi}{\kappa} \chi(M)$, and it is in this sense that two-dimensional gravity is a purely topological theory. \\ 

If we were to proceed in a similar manner for a non-Hausdorff action, then it would seem natural to maintain the topological nature of the theory and define the action according to a Euclidean version of the non-Hausdorff Gauss-Bonnet theorem. For a non-Hausdorff manifold $\textbf{M}$ constructed via Theorem \ref{THM: smooth nh manifold} and endowed with a Riemannian metric $\textbf{h}$, we write:
\begin{equation}\label{EQ: nh Euclidean action}
    \mathcal{S}(\textbf{M}, \textbf{h}) := \frac{1}{2\kappa}\int_{\textbf{M}} R dA + \frac{1}{\kappa}\int_{\partial \textbf{M}} k d\gamma - \frac{1}{\kappa}\int_{\textsf{Y}} k d\gamma.  
\end{equation}  
By construction, the set $\textsf{Y} = \overline{A}\backslash A$ is precisely one-half of the Hausdorff-violating submanifold sitting inside $\textbf{M}$. This can be seen as a sort-of ``interior surface term" or perhaps an ``internal boundary component", the dynamics of which are captured by an extra Gibbons-Hawking-York term of the appropriate sign. \\

Aside from a recreating a topological theory, the action $\mathcal{S}$ may also be justified according to a variational principle. Indeed, if we were to vary $\mathcal{S}$ with respect to the Riemannian metric $\textbf{h}$, then we are left with the following: 
\begin{align*}
    \delta_\textbf{h} \mathcal{S} & = \delta_\textbf{h}\left( \int_{\textbf{M}} R dA + \int_{\partial \textbf{M}} k d\gamma - \int_{\textsf{Y}} k d\gamma \right) \\
    & = \delta_\textbf{h}\Big{(} \int_{M_1} R dA + \int_{M_2} R dA - \int_{\overline{A}} R dA   + \int_{\partial M_1} k d\gamma + \int_{\partial M_2} k d\gamma - \int_{\partial A} k d\gamma - \int_{\textsf{Y}} k d\gamma\Big{)} \\
    & = \delta_{h_1} \left( \int_{M_1} R dA + \int_{\partial M_1} k d\gamma  \right) + \delta_{h_2} \left( \int_{M_2} R dA + \int_{\partial M_2} k d\gamma  \right) - \delta_{h_{\overline{A}}} \left( \int_{\overline{A}} R dA + \int_{\partial \overline{A}} k d\gamma  \right)\\
    & = 0 + 0 - 0 
\end{align*}
where here we have used the same delineations of the boundary as in Section 3.3 and have suppressed the overall factors of $\kappa^{-1}$ for readability. Observe that in the above expression we needed to include the additional geodesic curvature of the submanifold $\textsf{Y}$, since otherwise we are left with incomplete boundary data for the subspace $\overline{A}$.\\

Now, consider the non-Hausdorff trousers space $\textbf{T}$ as defined previously, but equipped with some Riemannian (rather than Lorentzian) metric $\textbf{h}$, defined according to our general discussion in Section 2.4. Topologically, $\textbf{T}$ consists of a pair of cylinders $M_i$ glued along a semi-open cylinder $A \cong S^1\times [0,1)$. These spaces all retract onto the circle, so their Euler characteristics vanish. According to Theorem \ref{THM: nH Euler characteristic}, the Euclidean action \eqref{EQ: nh Euclidean action} for $(\textbf{T}, \textbf{h})$ may then be evaluated as 
\begin{align*}
    S(\textbf{T},\textbf{h}) & = \frac{1}{2\kappa}\int_{\textbf{M}} R dA + \frac{1}{\kappa}\int_{\partial \textbf{M}} k d\gamma  - \frac{1}{\kappa}\int_{\textsf{Y}} k d\gamma \\ 
    & = \frac{2\pi}{\kappa} \chi(\textbf{T}) \stackrel{(\ref{THM: nH Euler characteristic})}{=} \frac{2\pi}{\kappa} (\chi(M_1) + \chi(M_2) - \chi(\overline{A})) = \frac{2\pi}{\kappa}(0 + 0 - 0) = 0.
\end{align*} 
In short, this means that the non-Hausdorff trousers space is flat, which agrees with the fact that $(\textbf{T}, \textbf{h})$ is locally isometric to the flat cylinder by construction. We may conclude from this that in the Euclidean theory, the non-Hausdorff trousers space will have a larger contribution to the path integral than the Euclidean version of the Hausdorff trousers space. \\

In fact, if we consider a broader path integral that sums over all topologies interpolating between all types of boundary components, then we are left with the observation that the non-Hausdorff trousers space contributes to the Euclidean path integral as strongly as the cylinder. Moreover, it appears as though one can make arbitrarily-complicated configurations of non-Hausdorff cylinders that remain globally flat, e.g. by adding extra legs to the trousers, which will similarly contribute with equal strength to the path integral. It seems that the only way to ensure the suppression of non-Hausdorff trousers in Euclidean gravity (relative to the cylinder) would be to artificially introduce an extra term that tracks Hausdorff-violation:
$$  \exp\{ - \mathcal{S}(\textbf{M},\textbf{g}) \}\longrightarrow  \exp\{ - \mathcal{S}(\textbf{M}, \textbf{g}) - \alpha(\textbf{M}) \},      $$
where here $\alpha$ ought to vanish for Hausdorff manifolds and be strictly positive otherwise. 

\subsubsection{Lorentzian Gravity}

In Lorentzian gravity, we would like to define the action as in \eqref{EQ: nh Euclidean action}, where here we compute the curvature and volume form according to the Lorentzian metric instead. We may again motivate this choice by the topological nature of the action, and again the variation of this action will vanish provided we include the additional Gibbons-Hawking-York term for the Hausdorff-violating submanifold. However, in contrast to the Euclidean setting, we now have the additional subtlety that the corner terms of the Hausdorff-violating submanifold need to be computed using Lorentzian turning angles. \\

We will illustrate this for the non-Hausdorff trousers space $(\textbf{T}, \textbf{h})$. Here there are two special points in $\overline{A}$ that are used in the construction of $\textbf{T}$. These are $p$, the initial pointlike source of topology change, and $q$, which is the final point of topology change. Observe that there are two future-directed null rays connecting $p$ to $q$, as pictured in Figure \ref{FIG: A-bar}, which makes the subset $\textsf{Y}$ of $\overline{A}$ a boundary consisting of piecewise-geodesic null rays. Expanding out the Lorentzian version of equation \eqref{EQ: nh Euclidean action}, we have: 
\begin{align*}
    \mathcal{S}(\textbf{T},\textbf{g}) & = \frac{1}{2\kappa}\int_{\textbf{M}} R dA + \frac{1}{\kappa}\int_{\partial \textbf{M}} k d\gamma - \frac{1}{\kappa}\int_{\textsf{Y}} k d\gamma \\
    & = \frac{1}{2\kappa}\left( \int_{M_1} RdA + \int_{M_2} RdA - \int_{\overline{A}} R dA \right)  + \frac{1}{\kappa}\left( \int_{S^1} k d\gamma + \int_{S^1} k d\gamma + \int_{S^1} k d\gamma \right)  \\
    & \ \ \ - \frac{1}{\kappa}\left( \int_{\gamma_1} k d\gamma + \int_{\gamma_2} k d\gamma + \sum \theta_{ext} \right),
\end{align*}
where here the $\gamma_i$ denote the two null rays connecting $p$ to $q$ in $\overline{A}\subset M_1$. The total scalar curvature terms above will all vanish, since each of the spaces involved is a flat cylinder. Moreover, all of the geodesic curvature terms will also vanish, since the copies of $S^1$ on the boundary of $\textbf{T}$ are all spacelike geodesic, and both the $\gamma_i$ are null geodesics from $p$ to $q$ in $\overline{A}$. Thus the total action reduces to the term $\sum \theta_{ext}$, which here is a sum of the turning angles of $\textsf{Y} = \overline{A}\backslash A$ at $p$ and $q$. \\

Since both curves in $\textsf{Y}$ are null rays, we are left with a computation of turning angles using the null angular formulae of \eqref{EQ: m null, u spacelike}--\eqref{EQ: m,n null, timelike quadrant}. Assuming that we may select the sign of the $i\epsilon$-regulariser for the tangent spaces $T_p \overline{A}$ and $T_q\overline{A}$ of the corners $p$ and $q$ independently, we are left with four possible combinations. These are:
\begin{enumerate}
    \item[(i)] $T_p \overline{A}$ and $T_q\overline{A}$ both use regulariser $+i\epsilon$,
    \item[(ii)] $T_p \overline{A}$ and $T_q\overline{A}$ both use regulariser $-i\epsilon$,
    \item[(iii)] $T_p \overline{A}$ uses regulariser $+i\epsilon$ and $T_q \overline{A}$ uses regulariser $-i\epsilon$,
    \item[(iv)] $T_p \overline{A}$ uses regulariser $-i\epsilon$ and $T_q \overline{A}$ uses regulariser $+i\epsilon$.
\end{enumerate}

For the first two cases, the two corner contributions to $\mathcal{S}(\textbf{T}, \textbf{g})$ have equal strength but opposite sign, so they cancel and the total action reduces to zero. Put differently, cases (i) and (ii) will correspond to global Wick rotations for a triangulation of $\overline{A}$, and thus fall within the scope of the Gauss-Bonnet theorems \ref{THM: Combinatorial LGBT} and \ref{THM: smooth LGBT}. The Euler characteristic of $\textbf{T}$ being zero then ensures that $\mathcal{S}(\textbf{T}, \textbf{g})$ vanishes, and we again encounter the same issue as in the Euclidean case. \\

Alternatively, we may also consider cases (iii) and (iv): if we were to pick regularisers with opposite signs, the sum of corner terms would not cancel. Thus we may \textit{not} apply the Gauss-Bonnet theorem in order to evaluate the action. Instead, we are left with actions taking the values 
\begin{equation}\label{EQ: NH actions for trousers space}
    \mathcal{S}^{(iii)}(\textbf{T}, \textbf{g}) = + \frac{\pi i}{\kappa}  \ \ \text{and} \ \  \mathcal{S}^{(iv)}(\textbf{T}, \textbf{g}) = - \frac{\pi i}{\kappa}.  
\end{equation}
The former will lead to a suppression of the non-Hausdorff trousers space in the schematic path integral \eqref{EQ: schematic Lorentzian path integral}, and the latter will lead to an enhancement. It seems reasonable, therefore, to suggest that we select regularisers of opposite sign, according to case (iii) above. For future reference, we state this here in more general form.  
\begin{equation}\label{EQ: Angular convention}
    \text{sgn}\big(i\epsilon(x)\big) = \begin{cases}
        +i\epsilon & \text{if} \ x \ \text{corresponds to an initial point of topology change} \\
        -i\epsilon & \text{if} \ x \ \text{corresponds to a final point of topology change} 
    \end{cases}
\end{equation}

There is an obvious and intentional similarity between our angular prescription and that of Sorkin/Louko \cite{louko1997complex}. Perhaps of interest is the difference in overall value of the resulting action -- in \cite{louko1997complex} the action of the Hausdorff Trousers space is computed to be $\frac{2\pi i}{\kappa}$ (with appropriate sign convention), and here, we obtain an overall strength of $\frac{\pi i}{\kappa}$. Thus, if simultaneously included within the same path integral, it appears as though the non-Hausdorff Trousers space will enjoy a weaker suppression factor relative to its Hausdorff counterpart.

\subsection{More $S^1$ Transitions}

As mentioned in the introduction, a generally unsung feature of Sorkin's angular convention is that any more transitions other than the Trousers spaces will yield a further dampening within the path integral. Indeed, adding an extra genus (which topologically amounts to taking a connected sum with the torus) changes the Euler characteristic by: 
$$  \chi(M \# T^2) = \chi(M) + \chi(T^2) - \chi(S^2) = \chi(M) - 2,      $$
so adding any extra genera to the bulk of a trousers space will continually decrease its total contribution to the path integral. In this sense, according to the same-sign convention for $i\epsilon$, any more elaborate branching will be further suppressed in the path integral. Alternatively, if we were to pick the other sign of the regulator, then we are left with the observation that transitions from $\Sigma_1$ to $\Sigma_2$ with arbitrarily-many genera will have the highest probability of occurring. With this in mind, it makes sense to pick the sign that entails the suppression of these higher-genus transitions. \\

In the non-Hausdorff case, we still have this genus complication, as well as the possibility that the topology of $\Sigma_1$ may change from multiple spacelike-separated sources. We will now treat these two possible complications separately. 

\subsubsection{Transitions with extra pointlike sources}
For a collection $p_1, \cdots, p_n$ of spacelike separated points on the cylinder, we may construct a non-Hausdorff manifold by gluing two cylinders together along the open subset $$B := M_1 \Big{\backslash} \bigcup_{\alpha = 1}^n J^+(p_\alpha)$$
by taking $f$ to be the identity map. The result will be a manifold that is homeomorphic to the non-Hausdorff trousers space $\textbf{T}$, but not isometric to it. The only difference is with the causal structure of the Hausdorff-violating submanifold $\textsf{Y}$ -- we will obtain a null boundary with a sort-of ``zigzag" structure, as depicted in Figure \ref{FIG: A-bar for multiple points}. \\

When evaluating the action for this spacetime, we are again forced to consider some sort of angular convention for the turning angles of the extra null boundary $\textsf{Y}$. In this case, consistently applying the angular convention \eqref{EQ: Angular convention} will force an additional suppression of these spaces relative to $\textbf{T}$. Precisely, if $\textbf{T}_n$ is the non-Hausdorff trousers space arising from $n$-many pointlike sources, then under our sign convention $\mathcal{S}(\textbf{T}_n)$ will take value $n\pi i/\kappa$. Thus our more-elaborately branching models will naturally enjoy a larger suppression in the path integral \eqref{EQ: schematic Lorentzian path integral}. 

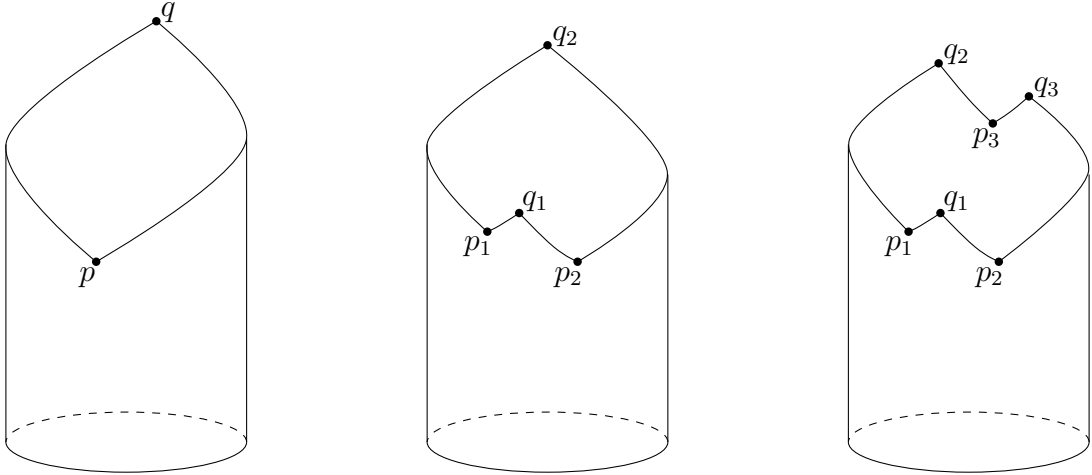
\begin{figure}
    \centering
\begin{tikzpicture}[scale=0.8]
    \draw (2,0) arc
	[start angle=360,	end angle=180, x radius=20mm, y radius =5mm] ;
\draw[dashed] (2,0) arc
	[start angle=0,	end angle=180, x radius=20mm, y radius =5mm] ;

 \draw[] (-2,0)--(-2,4.95);
 \draw[] (2,0)--(2,5.05);
\fill[] (-0.5,3) circle (0.7mm);
\fill[] (0.5,7) circle (0.7mm);

\node[] at (-0.65, 2.75) {$p$};
\node[] at (0.7, 7.15) {$q$};

\draw[] plot[smooth, tension=0.6] coordinates{(-0.5,3) (1.99,5) (0.5,7)};
\draw[] plot[smooth, tension=0.6] coordinates{(-0.5,3) (-1.99,5) (0.5,7)};


\draw (9,0) arc
	[start angle=360,	end angle=180, x radius=20mm, y radius =5mm] ;
\draw[dashed] (9,0) arc
	[start angle=0,	end angle=180, x radius=20mm, y radius =5mm] ;

 \draw[] (5,0)--(5,4.95);
 \draw[] (9,0)--(9,4.45);
 
\fill[] (7.5,3) circle (0.7mm);
\fill[] (6,3.5) circle (0.7mm);
\fill[] (6.53,3.81) circle (0.7mm);
\fill[] (7,6.6) circle (0.7mm);

\node[] at (7.35, 2.75) {$p_2$};
\node[] at (6.8,4) {$q_1$};
\node[] at (5.85,3.25) {$p_1$};
\node[] at (7.3,6.75) {$q_2$};

\draw[] plot[smooth, tension=0.6] coordinates{(6,3.5) (5.01,5) (7,6.6)};
\draw[] plot[smooth, tension=0.8] coordinates{(6,3.5) (6.2, 3.6) (6.53,3.81)};
\draw[] plot[smooth, tension=0.6] coordinates{(7.5,3) (8.99,4.5) (7,6.6)};
\draw[] plot[smooth, tension=0.8] coordinates{(7.5,3) (7.1, 3.25) (6.53,3.81)};


\draw (16,0) arc
	[start angle=360,	end angle=180, x radius=20mm, y radius =5mm] ;
\draw[dashed] (16,0) arc
	[start angle=0,	end angle=180, x radius=20mm, y radius =5mm] ;

 \draw[] (12,0)--(12,4.95);
 \draw[] (16,0)--(16,4.45);
 
\fill[] (14.5,3) circle (0.7mm);
\fill[] (13,3.5) circle (0.7mm);
\fill[] (13.53,3.81) circle (0.7mm);

\fill[] (13.5,6.3) circle (0.7mm);
\fill[] (14.4,5.3) circle (0.7mm);
\fill[] (15,5.75) circle (0.7mm);

\node[] at (14.35, 2.75) {$p_2$};
\node[] at (13.8,4) {$q_1$};
\node[] at (12.85,3.25) {$p_1$};

\node[] at (14.3,5.05) {$p_3$};
\node[] at (13.8,6.45) {$q_2$};
\node[] at (15.3,5.9) {$q_3$};

\draw[] plot[smooth, tension=0.6] coordinates{(13,3.5) (12.01,5) (13.5,6.3)};
\draw[] plot[smooth, tension=0.8] coordinates{(13,3.5) (13.2, 3.6) (13.53,3.81)};
\draw[] plot[smooth, tension=0.6] coordinates{(14.5,3) (15.99,4.5) (15,5.75)};
\draw[] plot[smooth, tension=0.8] coordinates{(14.5,3) (14.1, 3.25) (13.53,3.81)};

\draw[] plot[smooth, tension=0.8] coordinates{(14.4,5.3) (14.7, 5.5) (15,5.75)};
\draw[] plot[smooth, tension=0.8] coordinates{(14.4,5.3) (14, 5.7) (13.5,6.3)};
 
\end{tikzpicture}
\caption{The region $\overline{A}$ (left), together with more elaborate models. Based on Euler characteristic alone, there is no distinction between these three spaces. However, due to our angular convention \eqref{EQ: Angular convention} the total action will increase in multiples of $\pi i$. }
    \label{FIG: A-bar for multiple points}
\end{figure}

\subsubsection{Transitions with extra genera}
It is also possible that the non-Hausdorff trousers space may contain extra genera in the bulk, and these ought to be accounted for within the theory. In this situation we are combining both Morse-theoretic and Penrosian topology change, and there are many possible combinations to consider. Perhaps much more may be said about this particular regime, and there may well be limiting cases of interest, for instance, a non-Hausdorff branching point occurring precisely at the crotch singularity of a Trousers space. For simplicity's sake, we will assume that the future null cones of any pointlike sources of topology change do not contain any critical points such as the crotch singularity. Moreover, we will assume that $\textbf{M}$ is a non-Hausdorff manifold built from Hausdorff manifolds $M_1$ and $M_2$ glued along order-preserving isometries, so that the Gauss-Bonnet theorem of \ref{THM: Lorentzian non-Hausdorff Gauss-Bonnet} and the causal properties of Section 4.1 are still maintained.\footnote{Note that the desired poset structure and  global time function still exist, provided that we glue along isometries that preserve time-orientation. In such a situation, the causal orderings of the $M_i$ will be preserved under the canonical maps $\phi_i$, and we may transfer the global time functions of the $M_i$ to $\textbf{M}$ using Theorem \ref{THM: sections are fibre product}. } \\

We will now rewrite the Lorentzian non-Hausdorff action in such a manner that we obtain a purely topological piece, as well as a piece that spoils the Gauss-Bonnet theorem. The former will arise from the action evaluated with the global $+i\epsilon$ angular convention, while the latter will encode the effect of switching to our chosen angular convention \eqref{EQ: Angular convention}. \\

Let $\tilde\theta_{ext}$ denote the turning angles of $\textsf{Y}$ that would follow from the $+i\epsilon$ convention, and let $\theta_{ext}$ denote the angles that follow from our convention \eqref{EQ: Angular convention}. We can then write the $\textsf{Y}$ boundary term as:
\begin{align*}
    \int_{\textsf{Y}} k d\gamma & = \sum \int_{\gamma} k d\gamma + \sum \theta_{ext} \\
    & = \left( \sum \int_{\gamma} k d\gamma  + \sum\tilde{\theta}_{ext}\right) + \left( \sum \theta_{ext} - \sum \tilde{\theta}_{ext}  \right).
\end{align*}
Now, at each corner, the quantity $\theta_{ext} - \tilde{\theta}_{ext}$ is either $0$ or $-i\pi$, depending on whether the corner corresponds to an initial or final point of topology change, respectively. Thus, we have:
\begin{align*}
    \int_{\textsf{Y}} k d\gamma 
    & = \left( \sum \int_{\gamma} k d\gamma  + \sum\tilde{\theta}_{ext}\right) -n\pi i \ ,
\end{align*}
where $n$ is the number of pointlike sources of topology change. Plugging this into the action and using the Gauss-Bonnet theorem (which applies for the angles $\tilde\theta_{ext}$), we get:
\begin{align*}
  \mathcal{S}(\textbf{M}, \textbf{g}) &= \frac{1}{2\kappa}\int_{\textbf{M}} R dA + \frac{1}{\kappa}\int_{\partial \textbf{M}} k d\gamma - \frac{1}{\kappa}\int_{\textsf{Y}} k d\gamma \\
  &= \frac{1}{\kappa}\left(\frac{1}{2}\int_{\textbf{M}} R dA + \int_{\partial \textbf{M}} k d\gamma 
    - \sum \int_{\gamma} k d\gamma - \sum \tilde\theta_{ext}\right) + \frac{n\pi i}{\kappa} \\
  &= -\frac{2\pi i}{\kappa} \chi(\textbf{M}) + \frac{n \pi i}{\kappa} \ .
\end{align*}
We see that in distinction to our Euclidean discussion of Section 4.2, our convention \eqref{EQ: Angular convention} for Lorentzian turning angles at corners naturally introduces a term that suppresses non-Hausdorff topology changes, in addition to the usual suppression of Hausdorff ones (i.e. of higher genera) via the Euler characteristic. 

\section{Conclusion}

In this paper we have provided a basic analysis of non-Hausdorff transitions between various copies of the circle, in a manner consistent with the original idea of Penrose \cite{penrose1979singularities}. According to the colimit constructions of Section 2, we saw that non-Hausdorff manifolds may be readily endowed with a smooth structure, as well as all tensor fields that one may require in differential geometry. All of these notions were defined without issue according to local definitions together with consistency conditions on the map $f:A\rightarrow M_2$ which provided the specifics of the gluing construction. \\

A central issue for non-Hausdorff manifolds is the passing from local to global data. This manifests in the non-existence of arbitrary partitions of unity (cf. Theorem \ref{THM: no partitions of unity}), and in the inability to use the usual notion of integration. However, we circumvented this issue and saw in the integral formula of Theorem \ref{THM: Integration formula} that the global integral of a compactly-supported top form on a non-Hausdorff manifold satisfies a particular subadditivity formula. The formula \ref{THM: Integration formula} is almost identical to the standard subadditivity for Hausdorff manifolds, except that it required the crucial inclusion of the extra boundary term of $\partial \overline{A}$. According to the colimit construction of Section 2, this term integrates the extra piece of the Hausdorff-violating submanifold sat inside $\textbf{M}$. \\

Although innocuous from a measure-theoretic perspective, the inclusion of this ``internal boundary" term into the global integral had some important consequences for the non-Hausdorff Gauss-Bonnet theorem. Our main observation of Section 3 was an extension of the usual Gauss-Bonnet formula into the non-Hausdorff Lorentzian setting. Theorem \ref{THM: Lorentzian non-Hausdorff Gauss-Bonnet} ultimately showed that the total scalar curvature alone that does not equate to the Euler characteristic, but instead requires the extra integral counterterm, which in this context may be interpreted as the geodesic curvature of the extra Hausdorff-violating submanifold sitting inside the space. \\

As we saw in Section 4, this Gauss-Bonnet theorem ultimately suggests that there is no inherent way to guarantee the suppression of the non-Hausdorff trousers space within a Euclidean path integral that sums over topologies. However, in Lorentzian signature we had the additional subtlety that now, in order to properly evaluate the action, we needed to compute the particular turning angles between adjacent null geodesics. Given that each of these turning angles is subject to its own convention for Lorentzian angles, we argued that there indeed exists a method for suppression of non-Hausdorff trousers spaces in Lorentzian signature: we needed to intentionally spoil the Gauss-Bonnet theorem by selecting opposite sign conventions depending on the nature of these turning angles. Once this was done, the gravitational action took the correct sign in Lorentzian signature, which accounted the desired suppression. \\

In distinction to the Hausdorff case, there are two possible types of further branching: the addition of extra sources of pointlike change, and the inclusion of Hausdorff branching as per usual. In either case, we saw that more elaborate branching caused a further dampening in the path integral. 

\subsubsection*{Future Work}
We will now finish with some discussions of future work. As we saw in Sections 1.4 and 3.1, the non-Hausdorff Trousers space naturally inherits any Lorentzian metric placed on the ordinary cylinder. Moreover, we saw in Section 4.1 that the causal structure arising from the metric $\textbf{g}$ on $\textbf{T}$ is relatively well-behaved, thus on causal grounds alone one may suggest that the non-Hausdorff trousers space should be reasonably included in the schematic path integral of Equation \eqref{EQ: schematic Lorentzian path integral}. However, it seems natural to suggest that a``physically reasonable" spacetime ought to be a suitable background upon which to define both spinors and quantum fields. \\

Regarding possible spin structures: it is well-known in the Hausdorff case that questions regarding the existence and uniqueness of spin structures on a given manifold may be best articulated within the language of Čech cohomology. Specifically, it can be shown that a given (Hausdorff) manifold admits a spin structure if and only the second Steifel-Whitney number $w_2$ vanishes. In addition, it can be shown that the number of inequivalent spin structures on a given manifold may be classified with the Čech cohomology group $\check{H}^1(M, \mathbb{Z}_2)$. For the non-Hausdorff case, the Čech cohomology is partially understood \cite{o2023vectorbun}, though a full theory of non-Hausdorff spin geometry is yet to exist. \\

Regarding quantum fields: in the non-Hausdorff case the causal irregularities of the crotch singularity do not exist, so the abnormalities present in the Trousers space (discussed in \cite{anderson1986does, manogue1988trousers, buck2017sorkin}) will probably not exist either. Moreover, it appears as though the categorical phrasing of non-Hausdorff manifolds detailed in Section 2 may readily be applied in conjunction with the locally-covariant algebraic quantum field theory introduced in \cite{brunetti2003generally} to construct at least \textit{some} class of quantum fields. However, it is not clear what the exact properties of such fields ought to be, and this should be an interesting avenue of inquiry. \\

Broadly speaking, our discussion in this paper suggest that Penrose's topology changing spacetimes may be appropriately furnished with the mathematical structures required for an inquiry into physics. From a causal perspective, these non-Hausdorff spacetimes appear to be better behaved than the ordinary Trousers space, ultimately due to the absence of any metric singularities. This good behaviour suggests that non-Hausdorffness may be a more appropriate model of topology change in Lorentzian signature. Once the spin geometry and quantum field theory of non-Hausdorff manifolds is well-understood, it would be very interesting to study the inclusion of these spacetimes within pre-existing theories that sum over topologies. In particular, our non-Hausdorff transitions may provide a geometric realization of interacting (i.e. splitting and joining) strings within Lorentzian signature. 

\section*{Acknowledgements}

This paper was made possible by funding received by the Quantum Gravity unit at the Okinawa Institute of Science and Technology. We'd like to extend thanks to Timothy Budd, Renate Loll, Annegret Burtscher, Seth Asante, Bianca Dittrich, and Daniele Oriti for the helpful discussions. 

\printbibliography

@article{jee1984gauss,
  title={Gauss-Bonnet formula for general {L}orentzian surfaces},
  author={Jee, Dzan Jin},
  journal={Geometriae Dedicata},
  volume={15},
  pages={215--231},
  year={1984},
  publisher={Springer}
}

@article{chern1963pseudo,
  title={Pseudo-Riemannian geometry and the {G}auss-{B}onnet formula},
  author={Chern, SS},
  journal={An. Acad. Brasil. Ci},
  volume={35},
  pages={17--26},
  year={1963}
}

@article{avez1963formule,
  title={Formule de {G}auss-{B}onnet-{C}hern en m{\'e}trique de signature quelconque},
  author={Avez, Andr{\'e}},
  journal={Revista de la Uni{\'o}n Matem{\'a}tica Argentina},
  volume={21},
  number={4},
  pages={191--197},
  year={1963},
  publisher={Uni{\'o}n Matem{\'a}tica Argentina}
}

@article{Dowker:1997hj,
author = "Dowker, Fay and Surya, Sumati",
title = "{Topology change and causal continuity}",
eprint = "gr-qc/9711070",
archivePrefix = "arXiv",
journal = "Phys. Rev. D",
volume = "58",
pages = "124019",
year = "1998"
}

@article{Sorkin:1985bh,
author = "Sorkin, R. D.",
title = "{On Topology Change and Monopole Creation}",
reportNumber = "Print-85-0951 (SYRACUSE)",
journal = "Phys. Rev. D",
volume = "33",
pages = "978--982",
year = "1986"
}

@book{rotman2013introduction,
  title={An introduction to algebraic topology},
  author={Rotman, Joseph J},
  volume={119},
  year={2013},
  publisher={Springer Science \& Business Media}
}

@article{alty1995generalized,
  title={The generalized {G}auss-{B}onnet-{C}hern theorem},
  author={Alty, LJ},
  journal={Journal of Mathematical Physics},
  volume={36},
  number={6},
  pages={3094--3105},
  year={1995},
  publisher={American Institute of Physics}
}

@article{law1992neutral,
  title={Neutral geometry and the {G}auss-{B}onnet theorem for two-dimensional pseudo-Riemannian manifolds},
  author={Law, Peter R},
  journal={The Rocky Mountain Journal of Mathematics},
  volume={22},
  number={4},
  pages={1365--1383},
  year={1992},
  publisher={JSTOR}
}

@article{jia2022complex,
  title={Complex, {L}orentzian, and {E}uclidean simplicial quantum gravity: numerical methods and physical prospects},
  author={Jia, Ding},
  journal={Classical and Quantum Gravity},
  volume={39},
  number={6},
  pages={065002},
  year={2022},
  publisher={IOP Publishing}
}

@article{sorkin2019lorentzian,
  title={Lorentzian angles and trigonometry including lightlike vectors},
  author={Sorkin, Rafael D},
  journal={ arXiv:1908.10022},
  year={2019}
}

@article{asante2023complex,
  title={Complex actions and causality violations: Applications to {L}orentzian quantum cosmology},
  author={Asante, Seth K and Dittrich, Bianca and Padua-Arg{\"u}elles, Jos{\'e}},
  journal={Classical and Quantum Gravity},
  volume={40},
  number={10},
  pages={105005},
  year={2023},
  publisher={IOP Publishing}
}

@article{louko1997complex,
  title={Complex actions in two-dimensional topology change},
  author={Louko, Jorma and Sorkin, Rafael D},
  journal={Classical and Quantum Gravity},
  volume={14},
  number={1},
  pages={179},
  year={1997},
  publisher={IOP Publishing}
}

@article{o2023nonHausAdj,
  title={non-{H}ausdorff manifolds via adjunction spaces},
  author={O'Connell, David},
  journal={Topology and its Applications},
  volume={326},
  pages={108388},
  year={2023},
  publisher={Elsevier}
}

@article{o2023vectorbun,
  title={Vector bundles over non-{H}ausdorff manifolds},
  author={O'Connell, David},
  journal={Topology and its Applications},
  pages={108982},
  year={2024},
  publisher={Elsevier}
}

@article{o2023deRham,
  title={A non-{H}ausdorff de {R}ham Cohomology},
  author={O'Connell, David},
  journal={ arXiv:2310.17151},
  year={2023}
}

@article{mccabe2005topology,
  title={The topology of branching universes},
  author={McCabe, Gordon},
  journal={Foundations of Physics Letters},
  volume={18},
  number={7},
  pages={665--676},
  year={2005},
  publisher={Springer}
}

@incollection{penrose1979singularities,
  title={Singularities and time-asymmetry},
  author={Penrose, Roger},
  booktitle={General relativity},
  year={1979}
}

@article{geroch1970domain,
  title={Domain of dependence},
  author={Geroch, Robert},
  journal={Journal of Mathematical Physics},
  volume={11},
  number={2},
  pages={437--449},
  year={1970},
  publisher={AIP}
}

@book{do2016differential,
  title={Differential geometry of curves and surfaces: revised and updated second edition},
  author={Do Carmo, Manfredo P},
  year={2016},
  publisher={Courier Dover Publications}
}

@article{sorkin1997forks,
  title={Forks in the road, on the way to quantum gravity},
  author={Sorkin, Rafael D},
  journal={International Journal of Theoretical Physics},
  volume={36},
  pages={2759--2781},
  year={1997},
  publisher={Springer}
}

@article{dowker2003topology,
  title={Topology change in quantum gravity},
  author={Dowker, Fay},
  journal={The future of theoretical physics and cosmology, eds. GW Gibbons, EPS Shellard and SJ Rankin, Cambridge Univ. Press, Cambridge, UK},
  pages={436--452},
  year={2003}
}

@article{hawking1978quantum,
  title={Quantum gravity and path integrals},
  author={Hawking, Stephen W},
  journal={Physical Review D},
  volume={18},
  number={6},
  pages={1747},
  year={1978},
  publisher={APS}
}

@book{milnor2015lectures,
  title={Lectures on the h-cobordism theorem},
  author={Milnor, John},
  volume={2258},
  year={2015},
  publisher={Princeton university press}
}

@article{reinhart1963cobordism,
  title={Cobordism and the {E}uler number},
  author={Reinhart, Bruce L},
  journal={Topology},
  volume={2},
  number={1-2},
  pages={173--177},
  year={1963},
  publisher={Pergamon}
}

@book{hawking2023large,
  title={The large scale structure of space-time},
  author={Hawking, Stephen W and Ellis, George FR},
  year={2023},
  publisher={Cambridge university press}
}

@article{minguzzi2019lorentzian,
  title={Lorentzian causality theory},
  author={Minguzzi, Ettore},
  journal={Living reviews in relativity},
  volume={22},
  number={1},
  pages={3},
  year={2019},
  publisher={Springer}
}

@article{bernal2003smooth,
  title={On smooth {C}auchy hypersurfaces and Geroch's splitting theorem},
  author={Bernal, Antonio N and S{\'a}nchez, Miguel},
  journal={arXiv preprint gr-qc/0306108},
  year={2003}
}

@article{harris1990causal,
  title={The causal boundary of the trousers space},
  author={Harris, Steven G and Dray, Tevian},
  journal={Classical and Quantum Gravity},
  volume={7},
  number={2},
  pages={149},
  year={1990},
  publisher={IOP Publishing}
}

@article{lehner2016gravitational,
  title={Gravitational action with null boundaries},
  author={Lehner, Luis and Myers, Robert C and Poisson, Eric and Sorkin, Rafael D},
  journal={Physical Review D},
  volume={94},
  number={8},
  pages={084046},
  year={2016},
  publisher={APS}
}

@article{buck2017sorkin,
  title={The {S}orkin-{J}ohnston state in a patch of the trousers spacetime},
  author={Buck, Michel and Dowker, Fay and Jubb, Ian and Sorkin, Rafael},
  journal={Classical and Quantum Gravity},
  volume={34},
  number={5},
  pages={055002},
  year={2017},
  publisher={IOP Publishing}
}

@article{borde1999causal,
  title={Causal continuity in degenerate spacetimes},
  author={Borde, A and Dowker, HF and Garcia, RS and Sorkin, Rafael D and Surya, Sumati},
  journal={Classical and Quantum Gravity},
  volume={16},
  number={11},
  pages={3457},
  year={1999},
  publisher={IOP Publishing}
}

@article{dray1991particle,
  title={Particle production from signature change},
  author={Dray, Tevian and Manogue, Corinne A and Tucker, Robin W},
  journal={General relativity and gravitation},
  volume={23},
  pages={967--971},
  year={1991},
  publisher={Springer}
}

@article{hartle1998generalized,
  title={Generalized quantum theory and black hole evaporation},
  author={Hartle, James B},
  journal={arXiv preprint gr-qc/9808070},
  year={1998}
}

@article{anderson1986does,
  title={Does the topology of space fluctuate?},
  author={Anderson, Arlen and DeWitt, Bryce},
  journal={Foundations of Physics},
  volume={16},
  pages={91--105},
  year={1986},
  publisher={Springer}
}

@article{manogue1988trousers,
  title={The trousers problem revisited},
  author={Manogue, Corinne A and Copeland, ED and Dray, Tevian},
  journal={Pramana},
  volume={30},
  pages={279--292},
  year={1988},
  publisher={Springer}
}

@article{cheeger1984curvature,
  title={On the curvature of piecewise flat spaces},
  author={Cheeger, Jeff and M{\"u}ller, Werner and Schrader, Robert},
  journal={Communications in mathematical Physics},
  volume={92},
  number={3},
  pages={405--454},
  year={1984},
  publisher={Springer}
}

@article{brunetti2003generally,
  title={The generally covariant locality principle--a new paradigm for local quantum field theory},
  author={Brunetti, Romeo and Fredenhagen, Klaus and Verch, Rainer},
  journal={Communications in Mathematical Physics},
  volume={237},
  pages={31--68},
  year={2003},
  publisher={Springer}
}

@incollection{witten2022note,
  title={A note on complex spacetime metrics},
  author={Witten, Edward},
  booktitle={Frank Wilczek: 50 years of theoretical physics},
  pages={245--280},
  year={2022},
  publisher={World Scientific}
}

@inproceedings{sorkin1989consequences,
  title={Consequences of spacetime topology},
  author={Sorkin, Rafael D},
  booktitle={Proceedings of the Third Canadian Conference on General Relativity and Relativistic Astrophysics,(Victoria, Canada},
  pages={137--163},
  year={1989}
}

@article{de2022frontiers,
  title={Frontiers of Quantum Gravity: shared challenges, converging directions},
  author={de Boer, Jan and Dittrich, Bianca and Eichhorn, Astrid and Giddings, Steven B and Gielen, Steffen and Liberati, Stefano and Livine, Etera R and Oriti, Daniele and Papadodimas, Kyriakos and Pereira, Antonio D and others},
  journal={ arXiv:2207.10618},
  year={2022}
}

@book{belnap2021branching,
  title={Branching space-times: Theory and applications},
  author={Belnap, Nuel and M{\"u}ller, Thomas and Placek, Tomasz},
  year={2021},
  publisher={Oxford University Press}
}

@article{jubb2017boundary,
  title={Boundary and corner terms in the action for general relativity},
  author={Jubb, Ian and Samuel, Joseph and Sorkin, Rafael D and Surya, Sumati},
  journal={Classical and Quantum Gravity},
  volume={34},
  number={6},
  pages={065006},
  year={2017},
  publisher={IOP Publishing}
}

@article{feng2024singularity,
  title={Singularity at the demise of a black hole},
  author={Feng, Justin C and Mukohyama, Shinji and Carloni, Sante},
  journal={Physical Review D},
  volume={109},
  number={2},
  pages={024040},
  year={2024},
  publisher={APS}
}

@article{gerpenkron1972ideal,
  title={Ideal points in space-time},
  author={Geroch, Robert and Kronheimer, EH and Penrose, Roger},
  journal={Proceedings of the Royal Society of London. A. Mathematical and Physical Sciences},
  volume={327},
  number={1571},
  pages={545--567},
  year={1972},
  publisher={The Royal Society London}
}

@article{alty1995building,
  title={Building blocks for topology change},
  author={Alty, LJ},
  journal={Journal of Mathematical Physics},
  volume={36},
  number={7},
  pages={3613--3618},
  year={1995},
  publisher={American Institute of Physics}
}

@article{dowker1998handlebody,
  title={A handlebody calculus for topology change},
  author={Dowker, HF and Garcia, RS},
  journal={Classical and Quantum Gravity},
  volume={15},
  number={7},
  pages={1859},
  year={1998},
  publisher={IOP Publishing}
}

@article{neiman2013imaginary,
  title={Imaginary part of the gravitational action at asymptotic boundaries and horizons},
  author={Neiman, Yasha},
  journal={Physical Review D},
  volume={88},
  number={2},
  pages={024037},
  year={2013},
  publisher={APS}
}

@article{geroch1967topology,
  title={Topology in general relativity},
  author={Geroch, Robert P},
  journal={Journal of Mathematical Physics},
  volume={8},
  number={4},
  pages={782--786},
  year={1967},
  publisher={American Institute of Physics}
}

@book{hatcher2002algebraic,
  title={Algebraic Topology},
  author={Hatcher, Allen},
  year={2002},
  publisher={Cambridge University Press}
}

@article{OConnellthesis,
  title={Lorentzian Structures on Branching Spacetimes},
  author={O'Connell, David},
  journal={ILLC ePrints Archive},
  volume={MoL-2019-15},
  year={2019}
}

@article{luc2020interpreting,
  title={Interpreting non-{H}ausdorff (generalized) manifolds in General Relativity},
  author={Luc, Joanna and Placek, Tomasz},
  journal={Philosophy of Science},
  volume={87},
  number={1},
  pages={21--42},
  year={2020},
  publisher={The University of Chicago Press Chicago, IL}
}

@article{muller2013generalized,
  title={A generalized manifold topology for branching space-times},
  author={M{\"u}ller, Thomas},
  journal={Philosophy of science},
  volume={80},
  number={5},
  pages={1089--1100},
  year={2013},
  publisher={University of Chicago Press Chicago, IL}
}

@book{o1983semi,
  title={Semi-Riemannian geometry with applications to relativity},
  author={O'Neill, Barrett},
  volume={103},
  year={1983},
  publisher={Academic press}
}

@article{haefliger1957varietes,
  title={Vari{\'e}t{\'e}s (non s{\'e}par{\'e}es) {\`a} une dimension et structures feuillet{\'e}es du plan},
  author={Haefliger, Andr{\'e} and Reeb, Georges},
  journal={Enseign. Math.},
  volume={3},
  pages={107--126},
  year={1957}
}

@article{heller2011geometry,
  title={Geometry of non-{H}ausdorff spaces and its significance for physics},
  author={Heller, Michael and Pysiak, Leszek and Sasin, Wies{\l}aw},
  journal={Journal of mathematical physics},
  volume={52},
  number={4},
  pages={043506},
  year={2011},
  publisher={American Institute of Physics}
}

@article{hajicek1971causality,
  title={Causality in non-{H}ausdorff space-times},
  author={Hajicek, P},
  journal={Communications in Mathematical Physics},
  volume={21},
  number={1},
  pages={75--84},
  year={1971},
  publisher={Springer}
}

@book{lee2013smooth,
  title={Introduction to Smooth Manifolds},
  author={Lee, John M},
  year={2013},
  publisher={Springer}
}

@article{buss2012non,
  title={non-{H}ausdorff symmetries of {C}*-algebras},
  author={Buss, Alcides and Meyer, Ralf and Zhu, Chenchang},
  journal={Mathematische Annalen},
  volume={352},
  number={1},
  pages={73--97},
  year={2012},
  publisher={Springer}
}

@article{francis2023h,
  title={H-Unitality of Smooth Groupoid Algebras},
  author={Francis, Michael},
  journal={ arXiv:2307.00232},
  year={2023}
}

@article{hajicek1970extensions,
  title={Extensions of the {T}aub and {NUT} spaces and extensions of their tangent bundles},
  author={Hajicek, Petr},
  journal={Communications in Mathematical Physics},
  volume={17},
  pages={109--126},
  year={1970},
  publisher={Springer}
}

@article{crainic1999remark,
  title={A remark on sheaf theory for non-{H}ausdorff manifolds},
  author={Crainic, Marius and Moerdijk, Izak},
  year={1999}
}

@book{penrose1972techniques,
  title={Techniques of differential topology in relativity},
  author={Penrose, Roger},
  volume={7},
  year={1972},
  publisher={Siam}
}

@article{minguzzi2008causal,
  title={The causal hierarchy of spacetimes},
  author={Minguzzi, Ettore and S{\'a}nchez, Miguel},
  journal={Recent developments in pseudo-Riemannian geometry, ESI Lect. Math. Phys},
  pages={299--358},
  year={2008}
}

@article{yodzis1972lorentz,
  title={Lorentz cobordism},
  author={Yodzis, P},
  journal={Communications in Mathematical Physics},
  volume={26},
  number={1},
  pages={39--52},
  year={1972},
  publisher={Springer}
}

@article{yodzis1973lorentz,
  title={Lorentz cobordism {II}},
  author={Yodzis, P},
  journal={General relativity and gravitation},
  volume={4},
  number={4},
  pages={299--307},
  year={1973},
  publisher={Springer}
}

@article{wronski2009minkowskian,
  title={On {M}inkowskian branching structures},
  author={Wro{\'n}ski, Leszek and Placek, Tomasz},
  journal={Studies In History and Philosophy of Science Part B: Studies In History and Philosophy of Modern Physics},
  volume={40},
  number={3},
  pages={251--258},
  year={2009},
  publisher={Elsevier}
}

\end{document}